\documentclass[10pt, journal, final]{IEEEtran}
\usepackage{cite}
\usepackage{amsmath,amssymb,amsfonts,mathrsfs,bm,bbm}
\usepackage{enumitem}
\usepackage{upgreek}
\usepackage{graphicx}
\usepackage{textcomp}
\usepackage{subfig}
\usepackage{float}
\usepackage{xcolor}
\def\BibTeX{{\rm B\kern-.05em{\sc i\kern-.025em b}\kern-.08em
    T\kern-.1667em\lower.7ex\hbox{E}\kern-.125emX}}
    
\usepackage{tikz} 			   % for block diagrams
\usepackage{circuitikz}		   % for antennas
\usetikzlibrary{shapes}		   % for the cloud
\tikzstyle{block}=[draw, fill=white, minimum size=2em]
\usepackage{amsthm}             % for theorems
\usepackage{cuted}              % for two-columns equation

\newtheorem{theorem}{Theorem}  
\newtheorem{lemma}{Lemma}
\newtheorem{remark}{Remark}

\newtheorem{proposition}{Proposition}
\newtheorem*{corollary*}{Corollary}

\newcommand*{\herm}{\mathsf{H}}
\renewcommand{\vec}[1]{\mathbf{#1}}
\newcommand{\eqdef}{\stackrel{\Delta}{=}}

\newcommand{\dvec}[1]{\bm{#1}} 		% deterministic vectors
\newcommand{\rvec}[1]{\bm{#1}} 		% random vectors 
\newcommand{\dmat}[1]{\bm{\mathsf{#1}}}  	% deterministic vectors
\newcommand{\rmat}[1]{\mathbb{#1}} 	% random matrices
\newcommand{\mentry}[3]{\mathsf{#1}_{#2,#3}} % matrix entry
\newcommand{\EE}{\bm{\mathcal{E}}}		% expectation
\newcommand{\PP}{\mathrm{Pr}} 			% probability
\newcommand{\stdset}[1]{\mathbbmss{#1}}	% standard sets
\newcommand{\set}[1]{\mathcal{#1}}		% sets

\begin{document}

\title{Cooperative Multiple-Access Channels with Distributed State Information}

\author{Lorenzo~Miretti,~\IEEEmembership{Student~Member, IEEE,
} Mari~Kobayashi,~\IEEEmembership{Member, IEEE,
}\\David~Gesbert,~\IEEEmembership{Fellow, IEEE,
} Paul~de~Kerret~\IEEEmembership{Member, IEEE
}%
\thanks{Manuscript received November 19, 2019; revised June 7, 2020, October 30, 2020, and May 26,
2021; accepted June 1, 2021. This work was in part supported by the French-German Academy towards Industry 4.0 (SeCIF project) under Institut Mines-Telecom. L. Miretti, D. Gesbert, and P. de Kerret are also supported by the European Research Council under the European Union’s Horizon 2020 research and innovation program (Agreement no. 670896). This article was presented in part at the 2019 IEEE Information Theory Workshop. (\textit{Corresponding author: Lorenzo Miretti}). } 
\thanks{L. Miretti and D. Gesbert are with the Department of Communication Systems, EURECOM, Sophia Antipolis, France (e-mail: miretti@eurecom.fr;
gesbert@eurecom.fr).}
\thanks{M. Kobayashi is with the Department of Electrical and Computer Engineering, Technical University of Munich, Munich, Germany (e-mail:
mari.kobayashi@tum.de).}
\thanks{P. de Kerret was with the Department of Communication Systems, EURECOM, Sophia Antipolis, France; he is now with the Mantu Artificial Intelligence Laboratory (e-mail:
paul.dekerret@gmail.com)}}

\maketitle

\begin{abstract}
This paper studies a memoryless state-dependent multiple access channel (MAC) where two transmitters wish to convey a message to a receiver under the assumption of \textit{causal} and imperfect channel state information at transmitters (CSIT) and imperfect channel state information at receiver (CSIR). In order to emphasize the limitation of transmitter cooperation between physically distributed nodes, we focus on the so-called \textit{distributed} CSIT assumption, i.e., where each transmitter has its individual channel knowledge, while the message can be assumed to be partially or entirely shared a priori between transmitters by exploiting some on-board memory. Under this setup, the first part of the paper characterizes the common message capacity of the channel at hand for arbitrary CSIT and CSIR structure. The optimal scheme builds on Shannon strategies, i.e., optimal codes are constructed by letting the channel inputs be a function of current CSIT only. For a special case when CSIT is a deterministic function of CSIR, the considered scheme also achieves the capacity region of a common message and two private messages.
The second part addresses an important instance of the previous general result in a context of a cooperative multi-antenna Gaussian channel under i.i.d. fading operating in frequency-division duplex mode, such that CSIT is acquired via an explicit feedback of perfect CSIR. The capacity of the channel at hand is achieved by distributed linear precoding applied to Gaussian codes. Surprisingly, we demonstrate that it is suboptimal to send a number of data streams bounded by the number of transmit antennas as typically considered in a centralized CSIT setup. Finally, numerical examples are provided to evaluate the sum capacity of the binary MAC with binary states as well as the Gaussian MAC with i.i.d. fading.
\end{abstract}


\section{Introduction}
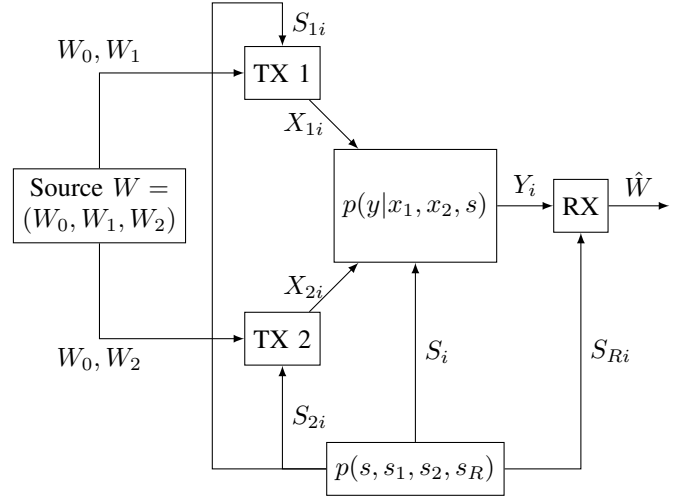
\begin{figure}[th!]
\centering
\begin{tikzpicture}[node distance=2.5cm,auto,>=latex]
		% Define blocks
		\node [block] (c) [draw,minimum size=1.5cm] {$p(y|x_1,x_2,s)$};	
    \node [block] (b1) [above left of=c] {TX 1};
		\node [block] (b2) [below left of=c] {TX 2};
    %\node [block] (a1) [left of=b1] {Src 1};
    %\node [block] (a2) [left of=b2] {Src 2};
    \node [block] (a) [left of=c, node distance=4.2cm, align=center] {Source $W = $\\$(W_0,W_1,W_2)$};
		\node [block] (d) [right of=c,  node distance=2.2cm] {RX};
    \node (end) [right of=d, node distance=1.3cm]{};
		\node (x1) [below of=c,coordinate] {};
		\node [block] (f) [below of = x1,node distance=1cm] {$p(s,s_1,s_2,s_R)$};

		% Define paths
    %\draw[->] (a1) -- node {$W_0,W_1$} (b1);
		%\draw[->] (a2) -- node {$W_0,W_2$} (b2);
		\draw[->] (a) |- node [swap]{$W_0,W_2$} (b2);
		\draw[->] (a) |- node {$W_0,W_1$} (b1);
    \draw[->] (b1) -- node [left]{$X_{1i}$} (c);
		\draw[->] (b2) -- node [left]{$X_{2i}$} (c);
    \draw[->] (c) -- node {$Y_i$} (d);
		\draw[->] (d) -- node {$\hat{W}$} (end);
		\draw[->] (f) --  ++(-2.7,0) -- ++(0,6.2) -| node [near end] {$S_{1i}$} (b1);
		\draw[->] (f) -| node [near end, swap] {$S_{2i}$} (b2);
		\draw[->] (f) -- node [swap]{$S_{i}$} (c);
		\draw[->] (f) -| node [near end, swap] {$S_{Ri}$} (d);
\end{tikzpicture}
\caption{Cooperative multiple-access channel with causal and distributed CSIT $S_1,S_2$ and imperfect CSIR $S_R$.} \label{fig:MAC}
\end{figure}
\IEEEPARstart{W}{ireless} communication networks can substantially benefit from transmitter cooperation via joint processing among multiple transmitters (TX), as it enables to mitigate interference and enhance the network performance. 
Although the benefits of TX cooperation have been identified in terms of coverage, throughput scaling, spectral efficiency, and energy efficiency, most of the existing cooperative schemes and performance analysis build on the common assumption that perfect, or at least perfectly shared, channel state information at the transmitters (CSIT), referred to {\it centralized} CSIT,  is available (see e.g. \cite{gesbert2010multi,simeone2012cooperative,quek2017cloud} and references therein). 
While such an assumption is convenient for analysis, it is however being challenged in a number of practical wireless scenarios. In fact, the acquisition of centralized CSIT always entails direct communication between transmitters or feedback from the receivers so that a given transmitter can collect the CSIT of other transmitters. This inevitably induces impairments and delays, which can be represented as a transmitter-specific distortion added to the channel state information. In order to capture such a limitation, we focus on the so-called \textit{distributed} CSIT such that each transmitter has its own channel knowledge. On the other hand, we assume that messages are partially or entirely shared between transmitters prior to the actual data transmission. This assumption is justified for instance in cache-aided networks where parts of delay-tolerant web content (video files) can be pre-fetched at transmitters typically during off-peak hours \cite{maddah2014fundamental,naderializadeh2019cache}. More generally, our setup can be thought as representing service situations where the CSIT time sensitivity is high in relation to that of the data contents.

By taking into account both practical CSIT limitation and inherent message-sharing opportunity, we study a state-dependent multiple access channel (MAC) illustrated in Fig.~\ref{fig:MAC}. Namely, two transmitters with respective state knowledge $S_1, S_2$ wish to cooperatively convey a message $W$ through inputs $X_{1}, X_2$ to a single receiver (RX) with state knowledge $S_R$.  
We do not consider direct communication links between transmitters that enable further online interactions such as conferencing \cite{WillemsThesis,permuter2011message}. Rather,  we aim to design the transmission strategy for a predefined CSI distribution mechanism described by the joint distribution of $(S,S_1,S_2,S_R)$ and a given message cooperation defined by the rate of $(W_0,W_1,W_2)$. More precisely, the {\it a priori} cooperation among the TXs is modeled by the following two components: % interacting forms:
\begin{itemize}
    \item \textit{State} cooperation, modeled by the distribution  of the CSIT $(S_1,S_2)$. Perfect state cooperation corresponds to a centralized CSIT configuration where the TXs share the same state information, i.e., $S_1=S_2$ (note that this does not necessarily correspond to perfect CSIT).
    \item \textit{Message} cooperation, modeled by splitting $W$ into $3$ sub-messages $(W_0,W_1,W_2)$, where $W_0$ is the portion of $W$ available at both TXs, and where $W_k$, $k=1,2,$ is the portion available only at TX~$k$. Perfect message cooperation corresponds to $W_0 = W$ or $W_1=W_2=\emptyset$.
\end{itemize}
Distributed CSIT gives rise to many interesting, yet challenging, problems because TXs must cooperate on the basis of uncertainties about each other's state information. There are roughly two classes of works. The first class 
focuses on signal processing methods \cite{gesbert2018team} such as particular precoders optimization, and asymptotic ergodic rate analysis in the regime of high signal-to-noise ratio \cite{dekerret2012degrees,bazco2019achieving} for cooperative multi-user networks with interference. The second class is based on the information theoretic models. These include the MAC with partial CSIT $S_1, S_2$ and full CSIR $S=S_R=(S_1, S_2)$\cite{hwang2007multiple}, the slow-fading Gaussian MAC with partial CSIT $S_1, S_2$ \cite{hwang2007multiple}, the MAC with conferencing encoders under noncausal partial CSIT $S_1, S_2$ and full CSIR $S=S_R=(S_1, S_2)$ \cite{permuter2011message}, the MAC with partial and strictly causal state information $S_k$ at TX $k$ and no CSIR \cite{li2012multiple}, and the cooperative MAC with non-causal CSIT at one TX and strictly causal at the other \cite{zaidi2013capacity}. Although useful system design and performance analysis are obtained from these two frameworks individually, they are disconnected each other in the sense that insights obtained from one class cannot be useful for another. 

Motivated by such an observation, we wish to close the gap between these two approaches by designing a simple yet information theoretically optimal cooperative scheme under distributed channel state information. To the best of our knowledge, such a result was not reported before. In particular, we study the capacity of a common message and the capacity region of three messages over a memoryless state-dependent MAC with causal distributed CSIT. Before summarizing the main contributions of the current work in Section \ref{ssec:results}, we first review the existing results on coding with causal CSIT under various network models.

\subsection{Coding with causal CSIT} \label{ssec:literature}
In \cite{shannon1958channels}, Shannon characterized the capacity of a memoryless state-dependent point-to-point channel with causal state knowledge at the transmitter $S=S_T\in \mathcal{S}$ and no CSIR $S_R=\emptyset$. The capacity of the channel at hand can be alternatively given by \cite{el2011network}
\begin{equation*}\label{eq:shannon_eq}
C = \max_{p(u), \; x=f(u,s)}I(U;Y),
\end{equation*}
where $U\in \mathcal{U}$ is an auxiliary random variable of finite cardinality and independent of $S$, and $f$ is a deterministic function. 
Notice that $U$ can be seen as an index for the family of functions $\mathcal{S} \to \mathcal{X}$, also known as \textit{Shannon strategies}.
This result has a practical impact to the design of modern wireless communication systems as it suggests that capacity is achieved by encoding the message through a function $f$ 
depending only on the \textit{current} CSIT.

In the following, we briefly summarize the existing results exploiting Shannon strategies. Shannon strategies were generalized to more general setups with imperfect CSIT $S_T$ and imperfect CSIR $S_R$ \cite{caire1999capacity}, and to particular cases of state process with memory in \cite{caire1999capacity}. 
Shannon strategies were also extended to more general network models, including degraded broadcast channels \cite{steinberg2005coding,sigurjonsson2005multiple},  
degraded relay channels \cite{sigurjonsson2005multiple, skoglund2008relay}, as well as multiple access channels \cite{das2002capacities,steinberg2005coding,
sigurjonsson2005multiple,jafar2006capacity,como2011capacity,como2013noisy,lapidoth2013common,
lapidoth2013double,baruch2008cooperative,zaidi2013capacity,zaidi2014cooperative,li2012multiple}.
By focusing on the MAC literature, the capacity region of the state-dependent MAC was studied by Das and Narayan in \cite{das2002capacities}, where multi-letter formulas are given for very general channel and CSI models. Unfortunately the multi-letter expressions provide very little insights and cannot be easily computed. In contrast, single-letter expressions on an achievable rate region $(R_1, R_2)$ have been derived for the case of a common state $S_1=S_2=S$ in \cite{sigurjonsson2005multiple}.
% but only inner and outer bounds are given. 
When the two states $(S_1, S_2)$ are independent, Shannon strategies are proved to achieve the sum capacity \cite{jafar2006capacity}. In practical wireless 
systems operating in frequency division duplexing (FDD) mode,  it is typical to assume that the CSIT is a deterministic function of the CSIR as CSIT is acquired as an explicit feedback from the receiver. 
Under this condition, the full capacity region of the state dependent MAC has been characterized in \cite{jafar2006capacity} for independent states and 
then generalized to arbitrarily correlated states in \cite{como2011capacity}. For the case of degraded message sets as a special case of Fig.~\ref{fig:MAC} when $W_2=\emptyset$, the capacity region
has been characterized for the case of one-sided CSIT ($S_2=\emptyset$) in \cite[Section IV]{baruch2008cooperative}.

Interestingly, Shannon strategies are known to be suboptimal in general state-dependent MAC. In particular, Lapidoth and Steinberg demonstrated that Shannon strategies
fail to achieve some rate pairs in the state-dependent MAC with no CSIR for the case of common state \cite{lapidoth2013common} and the case of independent states \cite{lapidoth2013double}, causally or strictly causally available at the encoders. This is because block-Markov encoding can help the TXs to provide some CSIR at the RX by compressing and sending past CSIT cooperatively \cite{lapidoth2013common} or non-cooperatively\cite{lapidoth2013double}. Clearly, this strategy is not useful when the CSIT is already available at the RX, as in \cite{como2011capacity} and in parts of this work. The scheme proposed in \cite{lapidoth2013common,lapidoth2013double} have been further generalized in \cite{li2012multiple}, where the TXs compress and send the past codewords along with past channel states. The idea of sending the past codewords via block-Markov encoding has been proposed for the MAC with feedback \cite{cover1981achievable,willems1982feedback,WillemsThesis}, while the idea of sending the past state together with new messages was also considered in the simultaneous state and data communication (see e.g. \cite{choudhuri2013causal} and references therein) and in the cooperative MAC with strictly causal CSIT \cite{zaidi2014cooperative}.

\subsection{Contributions and Paper Outline}
\label{ssec:results}
This paper provides the following contributions: 
\begin{enumerate}
    \item We demonstrate that the common message capacity of the memoryless state-dependent MAC under distributed CSIT is achieved by Shannon strategies for any CSI distribution $p(s,s_1,s_2,s_R)$ in Theorem \ref{th:C_common}. The exact characterization of the common message capacity complements the existing results in \cite{jafar2006capacity,lapidoth2013double,lapidoth2013common}, which consider only two private messages, the results in \cite{baruch2008cooperative} which assumes $S_2 = \emptyset$ and $W_2 = \emptyset$, and it extends the single-user results in \cite{shannon1958channels,caire1999capacity} to two distributed TXs.
    \item For the special case when the CSIT of each TX $k$ is a deterministic function $q_k$ of the CSIR, i.e., $S_k = q_k(S_R)$ for $k=1,2$, we prove that Shannon strategies achieve the full capacity region on three messages $(W_0, W_1, W_2)$ in Theorem \ref{th:MAC_det}. This extends the existing result \cite{como2011capacity} to the case when a common message is present. 
    The contribution of Theorem \ref{th:MAC_det} lies in our converse proof based on a standard information inequality chain, which overcomes the technique used in  \cite{jafar2006capacity} restricted to the case of independent states while significantly simplifying the approach in \cite{como2011capacity}.    
    \item By specializing the model of Theorem 2, we establish the common message capacity of the multiple-input multiple-output (MIMO) Gaussian fading channel operating in frequency-division-duplexing (FDD) mode in Theorem~\ref{th:C_distributed_MIMO}. We demonstrate that distributed linear precoding over Gaussian codewords based on Shannon strategies is optimal. The difference with respect to centralized CSIT is that each TX shall choose its precoding vector as a function of its channel knowledge, rather than global CSIT. As a key ingredient for the achievability proof, we allow the number of data streams conveying $W_0$ to grow large up to a given upper bound that depends on the CSIT cardinalities. Furthermore, as a non-trivial extension, Theorem \ref{th:C_region_MIMO} characterizes the entire capacity region of the aforementioned channel. The converse proof exploits the underlying channel structure and functional dependencies, while achievability builds on superposition encoding and the distributed precoding technique developed for Theorem \ref{th:C_distributed_MIMO}.
    
\item By letting the number of precoded data streams be the maximum dimension, in Proposition \ref{cor:computation_C} and Proposition  \ref{cor:computation_wsr} we prove that the optimal distributed precoding design, belonging to the well-known class of difficult problems called \textit{team decision} problems \cite{gesbert2018team}, can be cast into a convex form. Moreover, in Proposition~\ref{cor:necessary} and Proposition~\ref{prop:tighter_bound} we provide a more in-depth analysis on the optimal number of common data streams. Surprisingly, we prove that the common wisdom of limiting the number of precoded data streams by the number of transmit antennas is strictly suboptimal under a distributed CSIT setup. This is in sharp contrast to the case of centralized CSIT.   
\end{enumerate}

The paper is organized as follows. In Section \ref{sec:MAC}, we provide the formal system model and the main results for general MAC with distributed CSIT. Section \ref{sec:MIMO} presents the results for the specific cooperative MIMO MAC at hand. The insights given by the above sections are then further illustrated via numerical examples in Section \ref{sec:examples}. For readability purposes, most of the proofs are moved to appendices.

\section{Cooperative Multiple-Access Channels with Causal and Distributed CSIT}\label{sec:MAC}
This section first provides the general channel model and the basic definitions adopted throughout this work. Then, we present 
general results on the cooperative multiple access channel (MAC) with causal and distributed CSIT illustrated in Fig.~\ref{fig:MAC}. Hereafter, we follow the standard notation  of \cite{el2011network}. 

\subsection{System Model and Problem Statement}
\label{ssec:system_model}
%, by following closely the classical conventions given for example in \cite{el2011network}.
\paragraph{Channel Model}
Consider the state-dependent MAC in Fig.~\ref{fig:MAC}, with a common message $W_0$, two private messages $W_1,W_2$, inputs $X_{1} \in \set{X}_1$, $X_{2} \in \set{X}_2$, output $Y \in \set{Y}$, state $S\in \set{S}$, memoryless channel law $p(y|x_1,x_2,s)$, \textit{distributed} CSIT $(S_1,S_2) \in \set{S}_1 \times \set{S}_2$, and imperfect CSIR $S_R \in \set{S}_R$. The sequence of tuples $\{(S_i,S_{1i},S_{2i},S_{Ri})\}_{i=1}^n$ is assumed to follow a generic memoryless law $p(s,s_1,s_2,s_R)$.  An $n$-sequence of inputs, output and states is then governed by 
\begin{align*}
    p(y^n|x_1^n,x_2^n,s^n) &= \prod_{i=1}^np(y_i|x_{1i},x_{2i},s_i), \\
    p(s^n,s_1^n,s_2^n,s_R^n) &= \prod_{i=1}^np(s_i,s_{1i},s_{2i},s_{Ri}).
\end{align*}
We assume that three messages $W_0,W_1,W_2$ are independently and uniformly distributed over the sets $\set{W}_j \eqdef \{1,\ldots,2^{\lceil nR_j \rceil}\}$, $j=0,1,2$, where $R_j \geq 0$ is the rate of the message $W_j$. All alphabets are assumed to be finite, unless otherwise stated.

\paragraph{Encoding and Decoding} A block code $(2^{nR_0},2^{nR_1},2^{nR_2},n)$ of length $n$ with causal CSIT is defined by a set of encoding functions 
\begin{equation*}
    \phi_{ki}: \set{W}_0\times \set{W}_k\times \set{S}_k^i\to \set{X}_k, \quad k=1,2, \quad i=1,\ldots,n,
\end{equation*} 
yielding the transmitted symbols $x_{ki} = \phi_{ki}(w_0,w_k,s_k^i)$, as well as a decoding function 
\begin{equation*}
    \psi: \set{Y}^n\times \set{S}_R^n\to \set{W}_0\times \set{W}_1\times \set{W}_2,
\end{equation*} 
yielding the decoded messages $(\hat{w}_0,\hat{w}_1,\hat{w}_2) = \psi(y^n,s_R^n)$. Each encoder $k=1,2$ is subject to an average input cost constraint
\begin{equation*}
\EE\left[b^n(X_k^n)\right] \leq P_k,\quad b^n(x_k^n) \eqdef \dfrac{1}{n}\sum_{i=1}^n b_k(x_{ki}), 
\end{equation*}
where $b_k: \set{X}_k\to \mathbb{R}_+$ is a single-letter cost function upper-bounded by $b_{k,\max}$.
A rate-cost tuple $(R_0,R_1,R_2,P_1,P_2)$ is said to be achievable if, for the considered channel, there exists a family of block codes of length $n$ defined as before such that the average probability of error satisfies $P^{(n)}_e \eqdef \PP((\hat{W}_0,\hat{W}_1,\hat{W}_2)\neq(W_0,W_1,W_2))\rightarrow 0$ as $n\to \infty$. 

\paragraph{Figure of Merit}
For a given cost pair $(P_1,P_2)$, the closure of the set of all achievable rates $(R_0,R_1,R_2)$ is the capacity-cost region $\mathscr{C}(P_1,P_2)$ of the considered channel. In this work, we are mostly interested in two operating points in $\mathscr{C}(P_1,P_2)$. Namely, the common message capacity, defined by  
 $C_{0}(P_1,P_2) \eqdef \max\{R_0 \in \mathscr{C}(P_1,P_2)\}$ and the sum capacity $C_{\mathrm{sum}}(P_1,P_2) \eqdef \max\{R_{\rm sum}: (R_0,R_1,R_2)\in \mathscr{C}(P_1,P_2)\}$, where $R_{\mathrm{sum}} \eqdef R_0 + R_1 + R_2$ denotes the rate of the aggregate message $W\eqdef(W_0,W_1,W_2)$ in Fig.~\ref{fig:MAC}. 

\subsection{General Results}
As a non-trivial extension of \cite[Corollary 3]{baruch2008cooperative} with the one-sided CSIT ($S_2 = \emptyset$) and of \cite{shannon1958channels,caire1999capacity} for the centralized CSIT case $S_1 = S_2$ to a general CSI structure, we provide the main result of this section. 
\begin{theorem}\label{th:C_common}
The common message capacity of the channel in Fig.~\ref{fig:MAC} is given by
\begin{equation}\label{eq:C0}
    C_0(P_1,P_2) = \max_{\substack{p(u)\\x_k = f_k(u,s_k)\\\EE[b_k(X_k)]\leq P_k \; \forall k}}I(U;Y|S_R),
\end{equation}
where $U \in \set{U}$ is an auxiliary random variable of finite cardinality, independent of $(S,S_1,S_2,S_R)$, and where $f_k$ is a deterministic functions $\set{U}\times \set{S}_k\to \set{X}_k$ for $k=1,2$. 
\end{theorem}
\begin{IEEEproof}
The proof is given in Appendix \ref{sec:proof_th1}.
\end{IEEEproof}
\begin{remark}
By replacing $W_0$ by the triple $(W_0,W_1,W_2)$ in the converse proof of Theorem \ref{th:C_common}, it immediately follows that $C_{\mathrm{sum}}(P_1,P_2) \equiv C_0(P_1,P_2)$. 
\end{remark}

The main finding of Theorem \ref{th:C_common} is that the common message capacity (or equivalently, the sum capacity) can be achieved by Shannon strategies, i.e., by coding over the current CSIT $S_{1i},S_{2i}$ only while neglecting the past CSIT sequences. In fact, the converse proof also shows that providing the strictly causal sequence $(S_1^{i-1},S_2^{i-1})$ to both TXs does not increase the common message capacity. 

It is also worth emphasizing the difference with respect to the centralized CSIT where 
%Note that the expression in \eqref{eq:C0} is different than the expression for the common message capacity of a virtually centralized configuration, i.e. 
both TXs share $S_1 =S_2 \eqdef S_T$. In such case, by omitting for simplicity the input cost constraints, we recover in fact the classical result of \cite{shannon1958channels,caire1999capacity}
\begin{equation*}
    C_0 = \max_{\substack{p(u)\\(x_1,x_2) = f(u,s_T)}}I(U;Y|S_R).
\end{equation*}
Although Shannon strategies are optimal in both both distributed and centralized CSIT cases, the distributed CSIT assumption imposes the design of two different functions $f_1,f_2$ depending on the local CSIT only each, rather than a single $f$ as in the (virtually) centralized case.  

In order to prove the achievability part of Theorem \ref{th:C_common}, we obtain an achievable region for the MAC with a common message and two private messages as a byproduct. Specifically, we obtain the following result by combining Slepian-Wolf coding \cite{slepian1973coding} for the state-less MAC with common message and Shannon strategies \cite{shannon1958channels}. 
\begin{lemma}\label{lem:MAC}
For the channel in Fig.~\ref{fig:MAC}, $\mathscr{C}(P_1,P_2)$ includes the convex hull of all rate triples $(R_0,R_1,R_2)$ such that
\begin{align*}
	\begin{split}
    R_1 &\leq I(U_1;Y|U_2,U_0,S_R),\\
    R_2 &\leq I(U_2;Y|U_1,U_0,S_R),\\
    R_1+R_2 &\leq I(U_1,U_2;Y|U_0,S_R),\\
    R_0+R_1+R_2 &\leq I(U_1,U_2;Y|S_R),
    \end{split}
\end{align*}
for some auxiliary variables $(U_0,U_1,U_2) \in \set{U}_0\times \set{U}_1 \times \set{U}_2$ of finite cardinality, independent of the CSI $(S,S_1,S_2,S_R)$, with pmf factorizing as $p(u_0)p(u_1|u_0)p(u_2|u_0)$, and for some deterministic functions $f_k:\set{U}_k\times \set{S}_k\to\set{X}_k$, $x_k = f_k(u_k,s_k)$, satisfying $\EE[b_k(X_k)]\leq P_k$ for $k=1,2$.
\end{lemma}

It is well known that Shannon strategies, i.e., the scheme of Lemma \ref{lem:MAC}, fail to achieve $\mathscr{C}(P_1,P_2)$ for general $p(s,s_1,s_2,s_R)$, as observed for a special case of a common CSIT $S_1=S_2=S$ and no CSIR $S_R = \emptyset$ in \cite{lapidoth2013common}. This is because block-Markov encoding enables two TXs to compress past state information and send it as a common message to provide possibly useful CSIR to the RX. Nevertheless, Theorem \ref{th:C_common} shows that such a scheme based on block-Markov encoding is not necessary for achieving the sum capacity of the considered setup. Namely, provided that $R_0$ is large enough (in the worst case, equal to $C_{\mathrm{sum}}(P_1,P_2)$), the sum capacity is indeed achievable by the scheme in Lemma \ref{lem:MAC}. 

In the following we focus on the particular case where each CSIT is a deterministic function of CSIR. This assumption is highly relevant to frequency-division duplex (FDD) systems, where each transmitter acquires channel knowledge via an explicit quantized feedback from the receiver.
As a straightforward extension of \cite[Theorem 4]{como2011capacity} and \cite[Theorem 5]{jafar2006capacity} restricted to two private messages\footnote{\cite[Theorem 4]{como2011capacity} generalized the case of independent states $(S_1, S_2)$ to arbitrarily correlated states $(S_1, S_2)$.}, we characterize the capacity region for three messages as follows. 
\begin{theorem}\label{th:MAC_det}
By assuming that $S_1 = q_1(S_R)$ and $S_2 = q_2(S_R)$, where $q_1,q_2$ are two deterministic functions, the capacity region $\mathscr{C}(P_1,P_2)$ of the channel in Fig.~\ref{fig:MAC} is given by the convex hull of all rate triples satisfying
\begin{align}
    \begin{split}\label{eq:MAC_det}
    R_1 &\leq I(X_1;Y|X_2,U,S_R),\\
    R_2 &\leq I(X_2;Y|X_1,U,S_R),\\
    R_1+R_2 &\leq I(X_1,X_2;Y|U,S_R),\\
    R_0+R_1+R_2 &\leq I(X_1,X_2;Y|S_R),
    \end{split}
\end{align}
for some pmf $p(x_1|s_1,u)p(x_2|s_2,u)p(u)$, where $U\in \set{U}$ is an auxiliary variable of finite cardinality and independent of $(S,S_1,S_2,S_R)$, satisfying $\EE[b_k(X_k)]\leq P_k$ for $k=1,2$.
\end{theorem}
\begin{IEEEproof}
The proof is given in Appendix \ref{sec:proof_th2}.
\end{IEEEproof} 
Our main contribution lies in the converse proof, which solves the issue highlighted in \cite{jafar2006capacity} through an appropriate identification of the auxiliary variable $U$, thus allowing to greatly simplify the non-traditional, yet innovative, converse proof given by \cite{como2011capacity}. Note that Theorem \ref{th:MAC_det} refers to a setup where the RX is fully informed about $(S_1,S_2)$, hence there is no need for the TXs to convey $(S_1,S_2)$ through block-Markov schemes as in \cite{lapidoth2013common,lapidoth2013double}. In contrast to the general CSI setup in Lemma \ref{lem:MAC}, 
the private messages $W_1,W_2$ can be directly encoded into the input alphabets $\set{X}_1, \set{X}_2$ as observed already in \cite[Theorem 5]{jafar2006capacity}.
In light of Theorem \ref{th:MAC_det} we highlight the following expression for the common message capacity, which will be used to prove the main result of the second part of this paper.
\begin{remark}
Under the assumption $S_1 = q_1(S_R)$ and $S_2 = q_2(S_R)$, the expression in \eqref{eq:C0} is equivalently given by
\begin{equation}\label{eq:C_MAC_det}
   C_0(P_1,P_2) = \max_{\substack{p(u)\\x_k = f_k(u,s_k)\\\EE[b_k(X_k)]\leq P_k \; \forall k}} I(X_1,X_2;Y|S_R),
\end{equation}
where $U \in \set{U}$ is an auxiliary random variable of finite cardinality, independent of $(S,S_1,S_2,S_R)$, and where $f_k$ is a deterministic functions $\set{U}\times \set{S}_k\to \set{X}_k$ for $k=1,2$. 
\end{remark}

We conclude the first part of this paper by providing the following outer bound.
\begin{proposition}\label{prop:outer_cond_indep} Under the Markov chain $S_1 \to S_R \to S_2$, $\mathscr{C}(P_1,P_2)$ is included in the convex hull of all rate triples satisfying
\begin{align*}
    \begin{split}
    R_1+R_2 &\leq I(U_1,U_2;Y|U_0,S_R),\\
    R_0+R_1+R_2 &\leq I(U_1,U_2;Y|S_R),
    \end{split}
\end{align*}
for some auxiliary variables $(U_0,U_1,U_2) \in \set{U}_0\times \set{U}_1 \times \set{U}_2$ of finite cardinality, independent of the CSI $(S,S_1,S_2,S_R)$, with pmf factorizing as $p(u_0)p(u_1|u_0)p(u_2|u_0)$, and for some deterministic functions $f_k:\set{U}_k\times \set{S}_k\to\set{X}_k$, $x_k = f_k(u_k,s_k)$, satisfying $\EE[b_k(X_k)]\leq P_k$ for $k=1,2$. 
\end{proposition}
\begin{proof}
The proof is given in Appendix \ref{ssec:proof_outer_cond_indep}.
\end{proof}
Propostion \ref{prop:outer_cond_indep} shows that, for $(S_1,S_2)$ conditionally independent given $S_R$, Shannon strategies may be sub-optimal only in terms of individual rates (indeed, \cite{lapidoth2013double}  proves that higher individual rates are  achievable for independent $(S_1,S_2)$ and $S_R=\emptyset$). This extends \cite[Theorem 4]{jafar2006capacity}, which considered independent $(S_1,S_2)$ and no common message. The bound on $R_1+R_2$ was already reported in \cite {como2013noisy} and references therein by using the same technique as \cite{como2011capacity}. Similarly to Theorem \ref{th:MAC_det}, our contribution lies in a simpler converse proof.

\section{FDD Cooperative MIMO Channel with Fading}
\label{sec:MIMO}

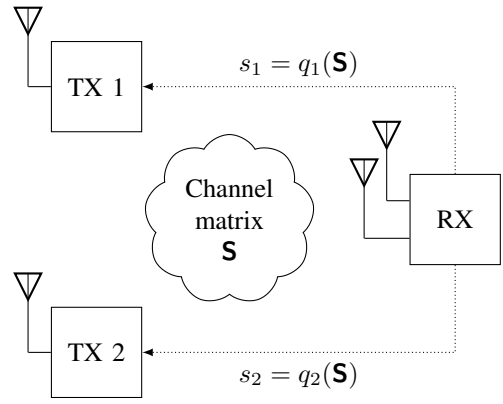
\begin{figure}[!t]
  \centering
\begin{tikzpicture}[node distance=2.5cm,auto,>=latex]
		% Define blocks
		\node [cloud, draw,cloud puffs=9,cloud puff arc=120, aspect=1, inner ysep=0em] (c) [minimum size=2cm, align=center] {Channel\\ matrix\\$\dmat{S}$};	
    \node [block] (b1) [minimum size=1.2cm, above left of=c] {TX 1};
		\node [block] (b2) [minimum size=1.2cm, below left of=c] {TX 2};
		\node [block] (d) [minimum size=1.2cm, right of=c,  node distance=3cm] {RX};

		% Define paths
		\draw[->,densely dotted] (d) |- node [near end, swap] {$s_1 = q_1(\dmat{S})$} (b1);
		\draw[->,densely dotted] (d) |- node [near end] {$s_2 = q_2(\dmat{S})$} (b2);
		
		% Antennas
		\draw[black, -] (b1.west)--++(0:-0.3cm) node[antenna, scale=0.5] {};
		\draw[black, -] (b2.west)--++(0:-0.3cm) node[antenna, scale=0.5] {};
		\draw[black, -] ([yshift=0.25cm]d.west)--++(0:-0.3cm) node[antenna, scale=0.5] {};
		\draw[black, -] ([yshift=-0.25cm]d.west)--++(0:-0.6cm) node[antenna, scale=0.5] {};
\end{tikzpicture}
  \caption{Illustration of a cooperative MIMO channel with fading state matrix $\dmat{S}$ known at the RX and distributed CSIT $s_1,s_2$ obtained from asymmetric feedback links.}
  \label{fig:scheme}
\end{figure}

In this section, we specialize the channel in Fig.~\ref{fig:MAC} to a practical $2\times2$ cooperative Gaussian MIMO channel with fading operating in FDD mode, illustrated in Fig.~\ref{fig:scheme}. The goal of this section is to particularize the general results of Section \ref{sec:MAC} and derive operational rules for encoding in the Gaussian MIMO setting. We point out that the results can be  generalized to more general antenna configurations and number of TXs as discussed in Section \ref{ssec:extension}.

\subsection{Channel Model and Notation}
\label{ssec:model_MIMO}
For this second part of the article, we extend our notations to better highlight multi-dimensional quantities. In particular, we denote deterministic matrices, column vectors and scalars by $\dmat{X}$, $\dvec{x}$, and $x$ respectively. We denote by $\mentry{X}{i}{j}$ the entry of $\dmat{X}$ in the $i$-th column and $j$-th row. Random matrices, column vectors and scalars are denoted instead by $\rmat{X}$, $\rvec{X}$, and $X$ respectively. The operators $( \cdot )^\herm$ and $\| \cdot \|_\mathrm{F}$ denote respectively the Hermitian transpose and the Frobenius norm. We denote by $\dvec{e}_i \in \{0,1\}^d$ a standard column selector, i.e., with the $i$-th element set to $1$ and all the other elements set to $0$. The identity matrix of dimension $n$ is denoted by $\dmat{I}_n$, or simply by $\dmat{I}$ when the dimension is clear from the context. Finally, $\stdset{C}$, $\stdset{R}_+$, and $\stdset{S}^n_+$ denote respectively the set of complex numbers, non-negative real numbers, and Hermitian positive-semidefinite matrices of dimension $n$.

We consider a classical Gaussian MIMO channel law $p(\dvec{y}|x_1,x_2,\dmat{S})$ and let the RX signal $\rvec{Y}\in \stdset{C}^{2\times 1}$ for a given channel use be given by:
\begin{equation*}
    \rvec{Y} = \rmat{S}\rvec{X}+\rvec{Z} = \rmat{S}\begin{bmatrix}X_1\\ X_2
    \end{bmatrix}+\rvec{Z},
\end{equation*}
where the state $\rmat{S} \in \stdset{C}^{2\times2}$ is a matrix of random fading coefficients, $X_k \in \stdset{C}$ is the signal transmitted by TX~$k$, subject to an average power constraint $\EE[|X_k|^2]\leq P_k$, and where $\rvec{Z} \sim \mathcal{CN}(\dvec{0},\dmat{I}_2)$ is independent of $\rmat{S}$. Furthermore, we focus on CSI distributions $p(\dmat{S},s_1,s_2,\dmat{S}_R)$  where the RX has perfect CSIR $\rmat{S}_R = \rmat{S}$, and where the CSIT is a quantized version of the CSIR, i.e., 
\begin{equation*}
    S_k = q_k(\rmat{S}), \quad k=1,2
\end{equation*}
\begin{equation*}
    q_k: \stdset{C}^{2\times2} \to \set{S}_k \eqdef \{1,\ldots,|\set{S}_k|\}, \quad |\set{S}_k| < \infty.
\end{equation*}
If $q_1 \neq q_2$, for example in the case of different feedback rates, we clearly fall into a distributed CSIT configuration. Note that, consistently with the definitions in Section \ref{ssec:system_model} with $S$ replaced by $\rmat{S}$, we consider coding over a large number of i.i.d. fading realizations. By using classical wireless terminology, we recall that this corresponds to the so-called \textit{fast-fading} regime.  

\subsection{Optimality of Distributed Linear Precoding with an Unconventional Number of Data Streams}
\label{ssec:MIMO_theorems}
In what follows, we establish the capacity region of the considered distributed setting and show that distributed linear precoding over Gaussian codewords is optimal. As we will see, the main novelty lies in an unconventional joint encoding technique for the common message $W_0$. For the sake of clarity, we present this technique by focusing on the common message capacity first.
\begin{theorem}\label{th:C_distributed_MIMO}
The common message capacity of the channel in Section \ref{ssec:model_MIMO} is given by $C_0(P_1,P_2) =$
\begin{equation}\label{eq:C_distributed_MIMO}
\max_{\substack{\dvec{g}_k(S_k)\in \stdset{C}^{d'} \\\EE\left[ \|\dvec{g}_k(S_k)\|^2\right]\leq P_k\; \forall k}}\EE\left[\log\mathrm{det}\left(\dmat{I}+\rmat{S}\dmat{\Sigma}(S_1,S_2)\rmat{S}^\herm\right)\right],
\end{equation}
where 
\begin{equation*}
    \dmat{\Sigma}(S_1,S_2) \eqdef \begin{bmatrix}
    \dvec{g}_1^\herm(S_1) \\ \dvec{g}_2^\herm(S_2) 
\end{bmatrix}\begin{bmatrix}
    \dvec{g}_1(S_1) & \dvec{g}_2(S_2) 
\end{bmatrix},
\end{equation*}
and where \begin{equation*}
    d'\leq d \eqdef |\set{S}_1|+|\set{S}_2|.
\end{equation*}
Furthermore $C_0(P_1,P_2)$ can be achieved by letting
\begin{equation}\label{eq:distr_precoding}
    \begin{bmatrix}
    X_1 \\ X_2 
\end{bmatrix}= \begin{bmatrix}
    \dvec{g}_1^\herm(S_1) \\ \dvec{g}_2^\herm(S_2) 
\end{bmatrix}\rvec{U},\quad \rvec{U}\sim \mathcal{CN}(\dvec{0},\dmat{I}_{d'}),
\end{equation}
where $\rvec{U}$ is the encoded common message.
\end{theorem}
\begin{IEEEproof}
We first apply the well-known \textit{maximum differential entropy lemma} \cite[p. 21]{el2011network} to the mutual information in  \eqref{eq:C_MAC_det} using the conditional input covariance $
    \dmat{\Sigma}(S_1,S_2)=\EE[\rvec{X}\rvec{X}^\herm|S_1,S_2]$, obtaining the upper bound 
\begin{equation*}
R_0 \leq \EE\left[\log\mathrm{det}\left(\dmat{I}+\rmat{S}\dmat{\Sigma}(S_1,S_2)\rmat{S}^\herm\right)\right].
\end{equation*}
For the achievability part, it suffices to show that every feasible $\dmat{\Sigma}(S_1,S_2)$ can be obtained via distributed linear precoders of dimension $d=|\set{S}_1|+|\set{S}_2|$. To this end, we first define the random vectors
\begin{equation*}
\tilde{\rvec{X}}_k \eqdef \begin{bmatrix}
f_k(U,1) \\ \vdots \\ f_k(U,|\set{S}_k|)
\end{bmatrix}, \quad k =1,2,
\end{equation*}
collecting the random inputs $X_k$ conditioned on each of the $|\set{S}_k|$ realizations of $S_k$, and the covariance matrix   
\begin{equation*}\label{eq:Q}
\dmat{Q} \eqdef \EE\left[\begin{bmatrix}
\tilde{\rvec{X}}_1 \\ \tilde{\rvec{X}}_2
\end{bmatrix} \begin{bmatrix}
\tilde{\rvec{X}}_1^\herm & \tilde{\rvec{X}}_2^\herm
\end{bmatrix}\right] \in \stdset{S}^{d\times d}_+.
\end{equation*}
It is easy to see that $\dmat{Q}$ contains all the elements of $\dmat{\Sigma}(s_1,s_2)$, $\forall (s_1,s_2) \in \set{S}_1\times \set{S}_2$. Note that, due to the power constraint $
    \sum_{s_k\in\set{S}_k}\EE[|f_k(U,s_k)|^2]p(s_k) \leq P_k < \infty$, 
any feasible $(f_1,f_2,p_U)$ must satisfy $\EE[|f_k(U,s_k)|^2] < \infty$ $\forall s_k \in \set{S}_k$, hence $\dmat{Q}$ has finite entries. Since $\dmat{Q}$ is Hermitian positive semi-definite, there exists a square matrix $\dmat{F}\in \stdset{C}^{d\times d}$ such that $\dmat{F}^\herm \dmat{F} = \dmat{Q}$. We denote its column vectors by
\begin{equation}\label{eq:F_matrix}
    \setlength\arraycolsep{3.5pt}
    \dmat{F} \eqdef \begin{bmatrix}\dvec{g}_1(1) & \ldots & \dvec{g}_1(|\set{S}_1|) & \dvec{g}_2(1) & \ldots & \dvec{g}_2(|\set{S}_2|) \end{bmatrix}.
\end{equation}
Finally, simple calculations show that the scheme in \eqref{eq:distr_precoding} with distributed linear precoders designed using the above procedure, i.e., selected from \eqref{eq:F_matrix}, preserves the desired $\dmat{\Sigma}(S_1,S_2)$, and  attains the maximum entropy upper bound. 
\end{IEEEproof}

The main result of Theorem \ref{th:C_distributed_MIMO} is that distributed linear precoding \cite{gesbert2018team} of shared Gaussian codewords achieves the performance limits of the considered cooperative MIMO setting. However, as a sufficient condition to prove achievability, Theorem \ref{th:C_distributed_MIMO} considers the transmission of possibly $d = |\set{S}_1|+|\set{S}_2|$ independent data streams. This unconventional design choice appears to be in sharp contrast with the centralized CSIT configuration (i.e., $S_1 = S_2 \eqdef S_T$), where the capacity of the $2\times 2$ MIMO channel is achieved by encoding $d'\leq 2$ streams in the presence of perfect message cooperation. In this latter case, by considering the per-antenna power constraint, the capacity takes the well-known expression given for example by \cite{caire2007adaptive}
\begin{equation}\label{eq:C_centralized}
    C_0(P_1,P_2) = \max_{\substack{\dmat{\Sigma}(S_T)\in\stdset{S}^2_+, \; \forall k\\\EE[\mentry{\Sigma}{k}{k}(S_T)]\leq P_k} }\EE\left[\log\mathrm{det}\left(\dmat{I}+\rmat{S}\dmat{\Sigma}(S_T)\rmat{S}^\herm\right)\right],
\end{equation}
where $\dmat{\Sigma}(s_T)\eqdef \EE[\rvec{XX}^\herm|S_T ]$ is the conditional input covariance. Clearly, the capacity in \eqref{eq:C_centralized} can be achieved by taking the matrix square-root $\dmat{G}(S_T) \eqdef \dmat{\Sigma}^{\frac{1}{2}}(S_T) \in \stdset{C}^{2\times2}$ and by letting
\begin{equation*}
    \begin{bmatrix}
    X_1 \\ X_2 
\end{bmatrix}= \dmat{G}^\herm(S_T)\rvec{U},\quad \rvec{U}\sim \mathcal{CN}(\dvec{0},\dmat{I}_2),
\end{equation*}
or, in other words, by precoding $d'=2$ data streams only.
Such an approach cannot be used for general distributed settings, as it generally leads to unfeasible linear precoders violating the functional dependencies $x_k=f_k(\dvec{u},s_k)$. 

The proof of Theorem \ref{th:C_distributed_MIMO} addresses this issue by \textit{increasing} the dimensionality $d'$ of the linear precoders up to $d \gg 2$, i.e., beyond conventional design choices. We now show that this unconventional linear precoding technique can be applied to extend Theorem \ref{th:C_distributed_MIMO} to the full capacity region.
\begin{theorem}\label{th:C_region_MIMO}
The capacity region $\mathscr{C}(P_1,P_2)$ of the channel in Section \ref{ssec:model_MIMO} is given by the union of all rate triples $(R_0,R_1,R_2)$ such that
\begin{align*}\label{eq:C_region_MIMO}
\begin{split}
	R_1 & \leq \EE\left[\log\left(1+\gamma_1(S_1)\|\rmat{S}\dvec{e}_1\|^2\right)\right]\\
	R_2 & \leq \EE\left[\log\left(1+\gamma_2(S_2)\|\rmat{S}\dvec{e}_2\|^2\right)\right]\\
	R_1 + R_2 &\leq \EE\left[\log\mathrm{det}\left(\dmat{I}+\rmat{S}\mathrm{diag}(\gamma_1(S_1),\gamma_2(S_2))\rmat{S}^\herm\right)\right]\\
    R_0+R_1+R_2 &\leq \EE\left[\log\mathrm{det}\left(\dmat{I}+\rmat{S}\dmat{\Sigma}(S_1,S_2)\rmat{S}^\herm\right)\right]
\end{split}
\end{align*}
where
\begin{equation*}
    \dmat{\Sigma}(S_1,S_2) = \begin{bmatrix}
    \dvec{g}_1^\herm(S_1) \\ \dvec{g}_2^\herm(S_2) 
\end{bmatrix}\left[
    \dvec{g}_1(S_1) \; \dvec{g}_2(S_2) 
\right] + \begin{bmatrix}
    \gamma_1(S_1) & 0 \\ 0 & \gamma_2(S_2) 
\end{bmatrix},
\end{equation*}
for some $\dvec{g}_k(S_k)\in \stdset{C}^{d'}$ and $\gamma_k(S_k) \in \stdset{R}_+$ such that 
\begin{equation*}
    d'\leq d \eqdef |\set{S}_1|+|\set{S}_2|
\end{equation*}
and $\EE\left[ \|\dvec{g}_k(S_k)\|^2\right]+\EE\left[\gamma_k(S_k)\right]\leq P_k$ for $k=1,2$.
Furthermore any point in $\mathscr{C}(P_1,P_2)$ can be achieved by letting
\begin{equation}\label{eq:achievable_scheme_MAC}
    \begin{bmatrix}
    X_1 \\ X_2 
\end{bmatrix}= \begin{bmatrix}
    \dvec{g}_1^\herm(S_1) & \sqrt{\gamma_1(S_1)} & 0 \\ \dvec{g}_2^\herm(S_2) & 0 & \sqrt{\gamma_2(S_2)}
\end{bmatrix}\begin{bmatrix}
\rvec{U} \\ V_1 \\ V_2
\end{bmatrix},
\end{equation}
where $(\rvec{U},V_1,V_2) \sim \mathcal{CN}(\vec{0},\dmat{I}_{d'+2})$ are the encoded common and private messages, respectively.
\end{theorem}
\begin{proof}
The proof is given in Appendix \ref{ssec:proof_capacity_region}.
\end{proof}
Theorem \ref{th:C_region_MIMO} shows that superposition of jointly and independently encoded Gaussian codes achieves the capacity region. However, while encoding of the private messages $W_1,W_2$ follows traditional approaches (one power-controlled stream per TX), the joint encoding of $W_0$ may require a larger number of precoded data streams (beyond two streams in our case) as already seen for Theorem \ref{th:C_distributed_MIMO}. A more detailed analysis of the role played by $d'$ in terms of optimal distributed precoding is provided in the following sections.

\subsection{Convex Reformulation for Optimal Distributed Precoding Design}
\label{ssec:computation}
The distributed precoding design problem \eqref{eq:C_distributed_MIMO} belongs to the class of \textit{static team decision} problems \cite{teams1972,gesbert2018team}, for which no efficient solutions are known in general. However, in contrast to the traditional precoding design with $d'\leq N_T$, where $N_T = 2$ is the total number of TX antennas, by letting $d'\leq d$ we are able to recast the optimal precoding design problem \eqref{eq:C_distributed_MIMO} into an equivalent \textit{convex} problem.
\begin{proposition}\label{cor:computation_C}
Problem \eqref{eq:C_distributed_MIMO} is equivalent to the following convex problem
    \begin{equation}\label{eq:C_distributedMIMO_convex}
        \begin{gathered}
         \underset{\dmat{Q} \in\stdset{S}^d_+}{\text{maximize}} \quad 
          \EE\left[ \log\mathrm{det}\left(\dmat{I}+\rmat{S}_{\mathrm{eq}}  \dmat{Q} \rmat{S}_{\mathrm{eq}}^\herm\right)\right] \quad \text{subject to} \\
          \sum_{i=1}^{|\set{S}_1|} \mentry{Q}{i}{i}\PP(S_1 = i) \leq P_1, \\
          \sum_{j=|\set{S}_1|+1} ^{|\set{S}_1|+|\set{S}_2|}\mentry{Q}{j}{j}\PP(S_2 = j - |\set{S}_1|) \leq P_2,
        \end{gathered}
    \end{equation}
    where we defined $\rmat{S}_{\mathrm{eq}} \eqdef \rmat{S}\dmat{E}^\herm(S_1,S_2) \in \stdset{C}^{2\times d}$, and
    \begin{equation*}
    \dmat{E}(S_1,S_2) \eqdef  \begin{bmatrix} 
    \dvec{e}_{S_1} & \bm{0} \\ 
    \bm{0} & \dvec{e}_{S_2}
    \end{bmatrix}\in \stdset{C}^{d\times 2},
    \end{equation*}
    where $\dvec{e}_{S_1} \in \{0,1\}^{|\set{S}_1|}$ and $\dvec{e}_{S_2} \in \{0,1\}^{|\set{S}_2|}$ are standard column selectors. 
\end{proposition}
\begin{IEEEproof} The proof follows by simply rewriting \eqref{eq:C_distributed_MIMO} in light of the technique used for the proof of Theorem~\ref{th:C_distributed_MIMO}. Specifically, we let $d'=d$ and use
$\rmat{S}\dmat{\Sigma}(S_1,S_2)\rmat{S}^\herm = \rmat{S}\dmat{E}^\herm(S_1,S_2) \dmat{F}^\herm \dmat{F} \dmat{E}(S_1,S_2)\rmat{S}^\herm = \rmat{S}_{\mathrm{eq}}  \dmat{Q} \rmat{S}_{\mathrm{eq}}^\herm$, where $\dmat{F}$ is given by \eqref{eq:F_matrix} and $\dmat{Q} = \dmat{F}^\herm \dmat{F}$. The power constraints correspond to linear constraints on the diagonal elements of $\dmat{Q}$.
\end{IEEEproof}
Problem \eqref{eq:C_distributedMIMO_convex} corresponds to the capacity of a virtual $d\times 2$ MIMO channel with state $\rmat{S}_{\mathrm{eq}}$, perfect CSIR, no CSIT, and (fixed) transmit covariance $\dmat{Q}$.
The capacity achieving distributed precoders for the original channel can be then designed from the optimal $\dmat{Q}^\star$ in Problem \eqref{eq:C_distributedMIMO_convex} as follows:
\begin{equation}\label{eq:square_root_design}
    \begin{bmatrix}
    \dvec{g}^\star_1(S_1) & \dvec{g}^\star_2(S_2) 
\end{bmatrix} = (\dmat{Q}^\star)^{\frac{1}{2}}\dmat{E}(S_1,S_2) \in \stdset{C}^{d\times 2}.
\end{equation}
An important remark here is that if the constraint $d'\leq d$ of Problem \eqref{eq:C_distributed_MIMO} is replaced by $d' < d$, the technique of  Proposition \ref{cor:computation_C} does not lead to a convex reformulation. This is because the $d\times d$ matrix $\dmat{F}$ is replaced by a $d'\times d$ matrix $\dmat{F}'$, hence introducing a non-convex constraint $\mathrm{rank}(\dmat{Q})\leq d' < d$ to Problem \eqref{eq:C_distributedMIMO_convex}. However, note that if the optimal $\dmat{Q}^\star$ for the unconstrained problem has rank $r < d$, then we can reduce with no loss of optimality the dimensionality of $\dvec{g}^\star_k(S_k)$ in \eqref{eq:square_root_design} down to $d'=r$.

Mirroring the previous section, the above result on common message capacity can be extended to the following weighted sum-rate maximization problem 
\begin{equation}\label{eq:wsr}
\underset{(R_0,R_1,R_2)\in \mathscr{C}(P_1,P_2)}{\text{maximize}}\alpha_0R_0+\alpha_1R_1+\alpha_2R_2,
\end{equation}
where $\alpha_k \geq 0$, $k=0,1,2$ are non-negative weights identifying rate priorities. We recall that the above problem can be used to characterize the boundary of $\mathscr{C}(P_1,P_2)$, since the weights can be interpreted as coefficients of a supporting hyperplane to such boundary (see e.g. \cite{el2011network,liu2006gausscommon} 
and references therein).
\begin{proposition}\label{cor:computation_wsr}
Problem \eqref{eq:wsr} is equivalent to the following convex problem
\begin{equation}\label{eq:wsr_convex}
\begin{gathered}
 \underset{\substack{R_k \in \stdset{R}_+,\;\dmat{Q} \in \stdset{S}^d_+ \\\gamma_1(S_1),\gamma_2(S_2)\in \stdset{R}_+}}{\text{maximize}} \quad
          \alpha_0R_1 + \alpha_1R_1+\alpha_2R_2   \quad \text{subject to}\\
 R_1  \leq \EE\left[\log\left(1+\gamma_1(S_1)\|\rmat{S}\dvec{e}_1\|^2\right)\right],  \\
		 R_2  \leq \EE\left[\log\left(1+\gamma_2(S_2)\|\rmat{S}\dvec{e}_2\|^2\right)\right],  \\
	 R_1 + R_2 \leq \EE\left[\log\mathrm{det}\left(\dmat{I}+\rmat{S}\mathrm{diag}(\gamma_1(S_1),\gamma_2(S_2))\rmat{S}^\herm\right)\right],  \\
	 \begin{aligned}
         & R_0 +R_1+R_2 \leq \\
         & \quad \EE\left[ \log\mathrm{det}\left(\dmat{I}+\rmat{S}_{\mathrm{eq}}  \dmat{Q} \rmat{S}_{\mathrm{eq}}^\herm + \rmat{S} \mathrm{diag}(\gamma_1(S_1),\gamma_2(S_2))\rmat{S}^\herm \right)\right],
         \end{aligned}\\
         \sum_{i=1}^{|\set{S}_1|}( \mentry{Q}{i}{i}+\gamma_1(i))\PP(S_1 = i) \leq P_1, \\
         \sum_{j=|\set{S}_1|+1} ^{|\set{S}_1|+|\set{S}_2|}(\mentry{Q}{j}{j}+\gamma_2(j-|\set{S}_1|))\PP(S_2 = j-|\set{S}_1|) \leq P_2,
\end{gathered}
\end{equation}
where $\rmat{S}_{\mathrm{eq}}$ is given as in Proposition \ref{cor:computation_C}.
\end{proposition}
\begin{IEEEproof}
The proof follows by the same technique as in the proof of Proposition \ref{cor:computation_C}. The details are omitted. 
\end{IEEEproof}

Problems \eqref{eq:C_distributedMIMO_convex} and \eqref{eq:wsr_convex} can be solved numerically via known convex optimization tools. A comprehensive discussion on the efficiency of various competing approaches is out of the scope of this work. Here, we point out two critical issues that should be taken into account in a practical system design. First, advanced stochastic optimization techniques may be required if the fading distribution $p(\dmat{S})$ is continuous\footnote{Note that all the results presented in this section do not require $\set{S}$ to be a discrete set, but only $|\set{S}_k|<\infty$ for $k=1,2$.}. Second, classical second-order methods as interior-point methods for semi-definite optimization typically scale badly with the dimension $d$ of $\dmat{Q}$. Hence, first-order methods may be more suitable whenever the cardinality of the CSIT alphabets $|\set{S}_k|$ is large. As a result of the algorithmic complexity stemming out of the above considerations (which are still very active research topics), we envision that feasible implementations of the proposed distributed precoding design should operate in an \textit{offline} fashion. Specifically, Problems \eqref{eq:C_distributedMIMO_convex} and \eqref{eq:wsr_convex} could be solved in a preliminary codebook design phase, while in the data transmission phase TX $k$ simply selects the precoder from the pre-designed codebook based on the received CSIT index $S_k$.

\subsection{Further Comments on the Optimal Number of Data Streams}
In this section we further elaborate on the optimal number of data streams $d'$ by addressing two important questions left open by Section \ref{ssec:MIMO_theorems}.
Theorem \ref{th:C_distributed_MIMO} shows that using a number of precoded data streams $d'= d$ is a sufficient condition for achievability of the common message capacity. However, we know that for some CSIT configurations (e.g., for centralized CSIT, as already discussed), $d'>2$ is not necessary. A first crucial question is whether there exists some distributed CSIT configuration for which such a condition is indeed necessary. In the next proposition we answer positively to this question.
\begin{proposition}
\label{cor:necessary}
For some $p(\dmat{S},s_2,s_2)$ and power constraints $(P_1,P_2)$, restricting $d'\leq N_T$ in problem \eqref{eq:C_distributed_MIMO}, where $N_T =2$ is the total number of TX antennas, leads to strictly suboptimal rates.
\end{proposition}
\begin{IEEEproof}
The proof is given in Appendix \ref{ssec:necessary_proof}, by showing the existence of a CSI distribution $p(\dmat{S},s_2,s_2)$ with binary CSIT $|\set{S}_1| = |\set{S}_2| = 2$ such that $d'\geq3$ is necessary for achieving $C_0(P_1,P_2)$.
\end{IEEEproof}

A second natural question is  whether the developed upper bound $d'\leq d = |\set{S}_1|+ |\set{S}_2|$ is tight, for some $p(\dmat{S},s_1,s_2)$. In the following we answer negatively to this question, by showing that indeed we can consider a slightly tighter upper bound. However, we firstly remark that obtaining tighter bounds is not trivial and is in fact related to the well-known low-rank matrix completion problem \cite{candes2010matrix}. Let us consider the matrix $\dmat{Q} \in \stdset{S}_+^d$ defined in the proof of Theorem \ref{th:C_distributed_MIMO}, or equivalently in Proposition \ref{cor:computation_C}, and its partition into blocks
\begin{equation}\label{eq:Q_structure}
\dmat{Q} = \begin{bmatrix}
\dmat{Q}_1 & \dmat{Q}_{12} \\
\dmat{Q}_{12}^\herm & \dmat{Q}_{2}
\end{bmatrix}, \quad \dmat{Q}_k \in \stdset{S}^{|\set{S}_k|}_+. 
\end{equation}
Informally, we recall that $\dmat{Q}$ collects the elements of the conditional input covariances $\dmat{\Sigma}(S_1,S_2)$ for all realizations of $(S_1,S_2)$.
By direct inspection of the capacity expression \eqref{eq:C_distributed_MIMO}, or equivalently of the objective in Problem \eqref{eq:C_distributedMIMO_convex}, we observe that the off-diagonal elements of the sub-matrices $\dmat{Q}_k$ do not contribute to the achievable rate, since they do not correspond to any element of any realization of $\dmat{\Sigma}(S_1,S_2)$.
Hence, by letting $\dmat{\tilde{Q}}$ be any optimal solution of \eqref{eq:C_distributedMIMO_convex}, the solution $\dmat{Q}^\star$ of the (non-convex) problem
\begin{equation}\label{eq:matrix_completion}
\begin{aligned}
        & \underset{\dmat{Q} \in\stdset{S}^d_+}{\text{minimize}}
        & & \mathrm{rank}(\dmat{Q}) \\
        & \text{subject to}
        & & \dmat{Q}_{12} = \dmat{\tilde{Q}}_{12} \\
        & & & \mentry{Q}{i}{i} = \mentry{\tilde{Q}}{i}{i}, \quad i = 1,\ldots,d
\end{aligned}
\end{equation}
is also an optimal solution of \eqref{eq:C_distributedMIMO_convex}, but where the off-diagonal elements of $\dmat{Q}_k$ have been optimized such that the rank is minimized. Since we have seen that the rank $r\leq d$ of $\dmat{Q}^\star$ corresponds to the dimension $d'$ of optimal distributed precoders (see Section \ref{ssec:computation}), establishing a tighter upper-bound on $d'$ can be cast into finding an upper-bound on the solution of \eqref{eq:matrix_completion}, which is an instance of a low-rank (semi-definite) matrix completion problem (see, e.g., \cite{candes2010matrix}).

To the best of the authors knowledge, non-trivial upper-bounds to problems of the type \eqref{eq:matrix_completion} remain elusive. Nevertheless, in the following proposition we provide a simple result showing the existence of a tighter upper-bound than $d = |\set{S}_1|+|\set{S}_2|$.
\begin{proposition}\label{prop:tighter_bound} The common message capacity $C_0(P_1,P_2)$ (resp. capacity region $\mathscr{C}(P_1,P_2)$) given by Theorem \ref{th:C_distributed_MIMO} (resp. Theorem \ref{th:C_region_MIMO}) is also achievable by letting  
\begin{equation*}
d' \leq |\set{S}_1|+|\set{S}_2| -1.
\end{equation*}
\end{proposition}
\begin{proof}
The proof is given in Appendix \ref{ssec:proof_tighter_bound}.
\end{proof}
The above result is by no means satisfactory, since the dimensionality reduction is marginal for large CSIT alphabets. Informally, the main limitation of the above bound is that the proof optimizes only one of the (coupled) variables in \eqref{eq:matrix_completion}. However, note that the above bound is tight for the toy example considered in the proof of Proposition \ref{cor:necessary}. 
% Attention, there are two consecutive textcolor paragraphs.

\subsection{Extension to Arbitrary Users and Antenna Configurations}
\label{ssec:extension}
Theorem \ref{th:C_distributed_MIMO} can be readily extended to $K$ TXs and arbitrary antenna configuration. By letting $N_k$ and $M$ be respectively the number of antennas at the $k$-th TX and at the RX, and by considering a fading state matrix $\rmat{S} \in \stdset{C}^{M\times (\sum_{k=1}^KN_k)}$ and distributed CSIT $(S_1,\ldots,S_K)$, $S_k = q_k(\rmat{S}) \in \set{S}_k$, $|\set{S}_k| < \infty$, it can be shown that the common message capacity is given by $C_0(P_1,\ldots,P_K) =$
\begin{equation}\label{eq:C0_multiantenna}
 \max_{\substack{\dmat{G}_k(S_k)\in \stdset{C}^{d'\times N_k}\\\EE\left[\|\dmat{G}_k(S_k)\|_{\mathrm{F}}^2\right]\leq P_k, \; \forall k}}\EE\left[\log\mathrm{det}\left(\dmat{I}+\rmat{S}\dmat{\Sigma}(S_1,\ldots,S_K)\rmat{S}^\herm\right)\right],
\end{equation}
where
\begin{equation*}
    \dmat{\Sigma}(S_1,\ldots,S_K) = \begin{bmatrix}
    \dmat{G}_1^\herm(S_1) \\ \vdots \\ \dmat{G}_K^\herm(S_K) 
\end{bmatrix}\begin{bmatrix}
    \dmat{G}_1(S_1) & \ldots & \dmat{G}_K(S_K) 
\end{bmatrix},
\end{equation*}
and where $d'\leq d \eqdef \sum_{k=1}^K N_k|\set{S}_k|$,
and it is achievable by distributed linear precoding of $d'$ i.i.d. Gaussian codewords.

The formal proof of the above statement follows from the same lines as for the proof of Theorem \ref{th:C_distributed_MIMO}, by considering $K$ functions $f_k: \set{U}\times \set{S}_k \to \stdset{C}^{N_k\times 1}$. The detailed steps do not provide additional intuitions, hence they are omitted. The convex reformulation of Proposition \ref{cor:computation_C} can be also similarly extended. Finally, as for Theorem \ref{th:C_region_MIMO}, the full-capacity region for $K$ private messages and a single common message can be achieved by superimposing standard MIMO MAC codes for $W_1,\ldots,W_K$ to the non-traditional distributed precoding technique achieving \eqref{eq:C0_multiantenna}.

\section{Numerical Examples}
\label{sec:examples}
\subsection{Channel with Additive Binary Inputs and State}
\label{ssec:num_example}
We consider the following channel
\begin{equation*}
    Y = X_1 + X_2 + S,
\end{equation*}
with binary inputs and state, i.e., $\set{X}_1 = \set{X}_2 = \set{S} = \{0,1\}$, and where $\set{Y}=\{0,1,2,3\}$. We do not consider input cost constraints. We further assume $S\sim \mathrm{Bernoulli}(q)$, no CSIR ($\set{S}_R=\emptyset$), and distributed CSIT $p(s_1,s_2|s) = p(s_1|s)p(s_2|s)$, where $p(s_k|s)$ is a binary symmetric channel with transition probability $\epsilon_k \in [0,0.5]$. Under the above model, the common message capacity (which coincides with the sum capacity) is given by 
\begin{equation*}
    C_0 = \max_{\substack{p(u)\\x_k = f_k(u,s_k)}}I(U;Y).
\end{equation*}
A (non-scalable) method for optimally solving the above optimization problem is to adapt to the considered distributed setting the original idea of coding over the \textit{alphabet of Shannon strategies} \cite{shannon1958channels,el2011network}, combined with classical results on the computation of the capacity of point-to-point channels \cite{cover2012elements}. More precisely, we proceed as follows:
\begin{enumerate}[leftmargin=*]
\item We build the alphabet of \textit{distributed Shannon strategies} by enumerating all the functions
\begin{equation*}
    t_u(s_1,s_2) = [t_{1,u}(s_1), t_{2,u}(s_2)],\quad t_{k,u}: \set{S}_k\to \set{X}_k,
\end{equation*}
where each function is indexed by $U$. There are $|\set{U}|=|\set{X}_1|^{|\set{S}_1|}|\set{X}_2|^{|\set{S}_2|} = 16$ such functions.
    \item We set $x_k = f_k(u,s_k) = t_{k,u}(s_k)$ and compute the equivalent state-less point-to-point channel 
    \begin{align*}
        &p(y|u) = \\
        &\sum_{\substack{x_1,x_2\\s,s_1,s_2}}&p(y|x_1,x_2,s)p(x_1|u,s_1)p(x_2|u,s_2)p(s,s_1,s_2).
    \end{align*}
    \item We run the Blahut-Arimoto algorithm for computing the capacity of the equivalent channel $p(y|u)$ \cite{cover2012elements}.
\end{enumerate}
Note that the above procedure is similar to the one outlined in \cite{keshet2008channel} for centralized settings. Furthermore, it can be readily generalized to arbitrary CSIR by simply considering an augmented output $\tilde{Y}\eqdef (Y,S_R)$.

In Fig.~\ref{fig:C_binary} we plot the capacity $C_0$ versus the CSIT quality at TX 2, for various choices of CSIT quality at TX 1, and for $q=0.5$. Note that $\epsilon_k = 0$ and $\epsilon_k = 0.5$ model respectively perfect and no CSIT at the $k$-th TX. Interestingly, the capacity of the system decreases with $\epsilon_2$ down to a flat regime in which any further decrease in quality does not matter, and the turning point depends on $\epsilon_1$. This can be interpreted as a regime in which the quality at one TX is so degraded that, although some CSIT is available, it does not allow for proper coordination with the better informed TX. Intuitively, it is important for the less informed TX to not act as unknown noise for the other TX. In fact, in the aforementioned regime it turns out that the optimal scheme at the less informed TX is to throw away completely its CSIT, making its behaviour not adaptive to the channel conditions but completely predictable by the more informed TX. 

\begin{figure}[t]
    \centering
    \includegraphics[width=0.8\columnwidth]{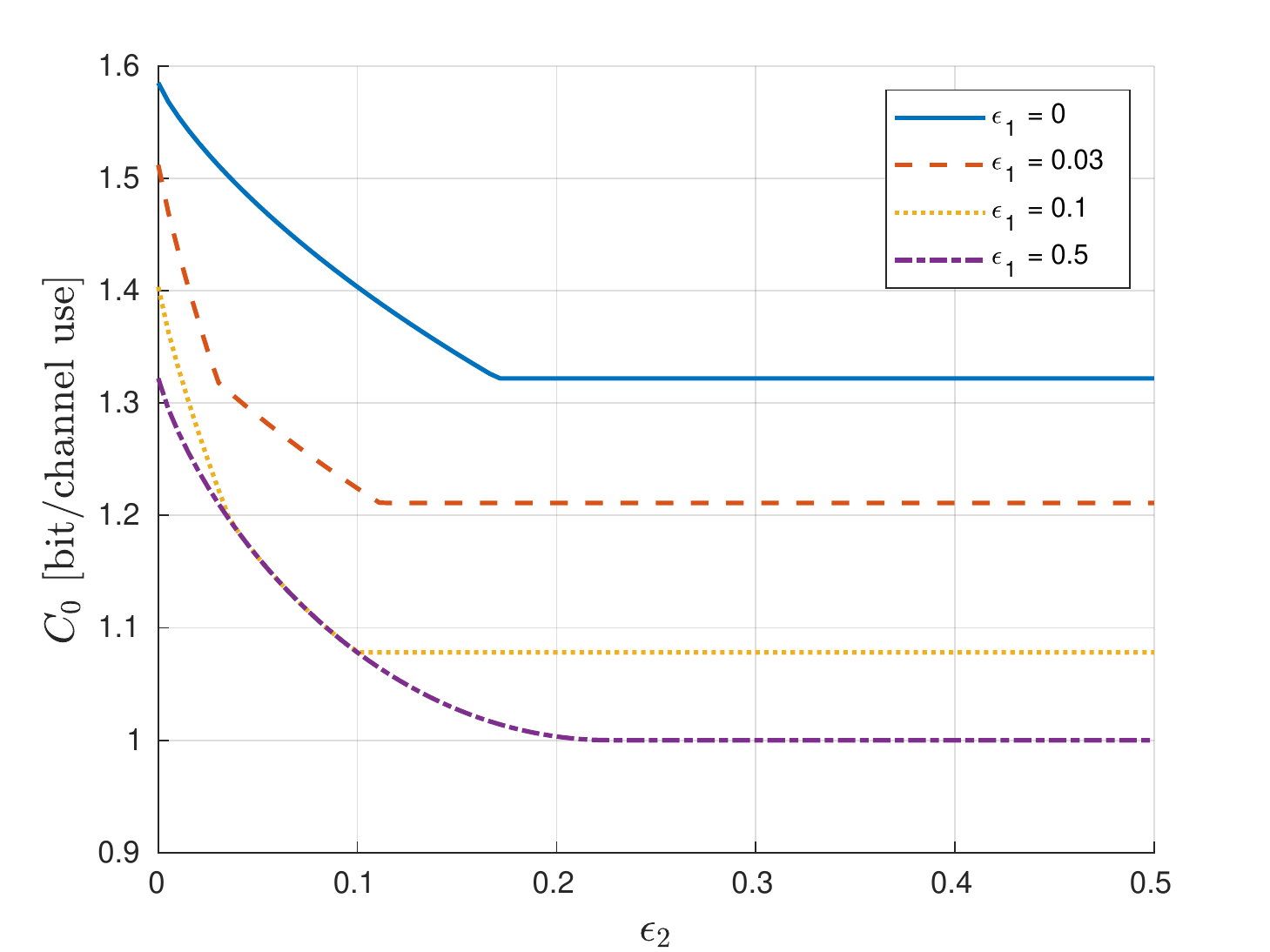}
    \caption{Capacity vs. CSIT distortion $\epsilon_2$ at TX 2, for various choices of CSIT distortion $\epsilon_1$ at TX 1.}
    \label{fig:C_binary}
\end{figure}

\subsection{Cooperative AWGN MIMO with Rayleigh Fading and Quantized Feedback}
In this section we simulate a practical cooperative MIMO channel with Rayleigh fading and with limited feedback rates. In particular, we let each element of $\rmat{S}$ to be i.i.d. $\mathcal{CN}(0,1)$, and we set for simplicity $P_1 = P_2 \eqdef \mathrm{SNR}$. The distributed CSIT configuration $p(\dmat{S},s_1,s_2)$ is given by two random quantizers with different rates $\beta_1,\beta_2$. 

More precisely, let $\set{S}_k = \{1,\ldots,2^{\beta_k}\}$ be the index set of a codebook $\{\hat{\dmat{S}}_{k,i} \}_{i=1}^{2^{\beta_k}}$ of randomly and independently generated codewords distributed as $p(\dmat{S})$. We then let $q_k(\dmat{S})$ to be a simple nearest neighbour vector quantizers in the Frobenius norm, i.e., $q_k(\dmat{S}) = \arg \min_{i \in \set{S}_k}\|\dmat{S}-\hat{\dmat{S}}_{k,i}\|_{\mathrm{F}}$. This scenario corresponds to an error-free feedback link from the RX to the $k$-th TX with limited rate of $\beta_k$ bits per channel realization. We set $\beta_1 = 4$ and $\beta_2 = 3$, which implies $d = |\set{S}_1|+|\set{S}_2| = 24$. We recall that the RX is assumed to have perfect CSIR.

We approximately solve Problem \eqref{eq:C_distributedMIMO_convex} through an off-the-shelf numerical solver for convex problems, by substituting $p(\dmat{S},s_1,s_2)$ with its empirical distribution 
$\hat{p}(\dmat{S},s_1,s_2) \eqdef \frac{1}{L}\sum_{l=1}^L\mathbbm{1}[(\dmat{S},s_1,s_2) = (\dmat{S}_{l},q_1(\dmat{S}_{l}),q_2(\dmat{S}_{l}))]$ obtained from $L= 1000$ i.i.d. samples $\{\dmat{S}_l\}_{l=1}^L$. This allows us to replace the expectation in \eqref{eq:C_distributedMIMO_convex} with a finite sum of $L$ convex functions. The capacity obtained is exact for a channel with state distribution equal to the empirical distribution $\hat{p}(\dmat{S},s_1,s_2)$, and approximates the capacity for ${p}(\dmat{S},s_1,s_2)$ as $L$ grows large. Furthermore, we repeat the above simulations by considering instead a single antenna at the RX, a setting denoted here as \textit{cooperative MISO}.

In Fig.~\ref{fig:C_sim} we plot the common message capacity versus SNR of a given instance of the considered channel model. We also plot the common message capacity for perfect CSIT at both TXs
\begin{equation*}
    C_0^{(\text{perf. CSIT})} = \max_{\substack{\dmat{\Sigma}(\rmat{S})\in\stdset{S}^2_+:\\ \EE[\mentry{\Sigma}{k}{k}(\rmat{S})] \leq P_k \; \forall k}} \EE\left[\mathrm{log} \mathrm{det} \left( \dmat{I}+\rmat{S}\dmat{\Sigma}(\rmat{S})\rmat{S}^\herm \right) \right], 
\end{equation*}
and for no CSIT
\begin{equation*}
    C_0^{(\text{no CSIT})} = \max_{\dmat{\Sigma}\in\stdset{S}^2_+: \; \mentry{\Sigma}{k}{k} \leq P_k \; \forall k}\EE\left[\mathrm{log} \mathrm{det} \left( \dmat{I}+\rmat{S}\dmat{\Sigma}\rmat{S}^\herm \right) \right].
\end{equation*}
We recall that these CSIT configurations are equivalent to a centralized $2\times 2$ MIMO system, hence we can simply use the classical MIMO results summarized e.g. in \cite{caire2007adaptive}, adapted to a per-antenna power constraint. For a fair comparison, these capacities are computed over the same empirical marginal distribution $\hat{p}(\dmat{S})$. As expected, for the MIMO case, the capacity gain given by distributed CSIT w.r.t no CSIT follows the well-known \textit{beamforming gain} trend of the perfect CSIT case, i.e., it vanishes in the high SNR regime. Similarly, for the MISO case, this gain converges to a constant power offset.

\label{ssec:MIMO_simulation}
\begin{figure}
    \centering
  \subfloat[$2\times 2$ Cooperative MIMO]{
       \includegraphics[width=0.8\columnwidth]{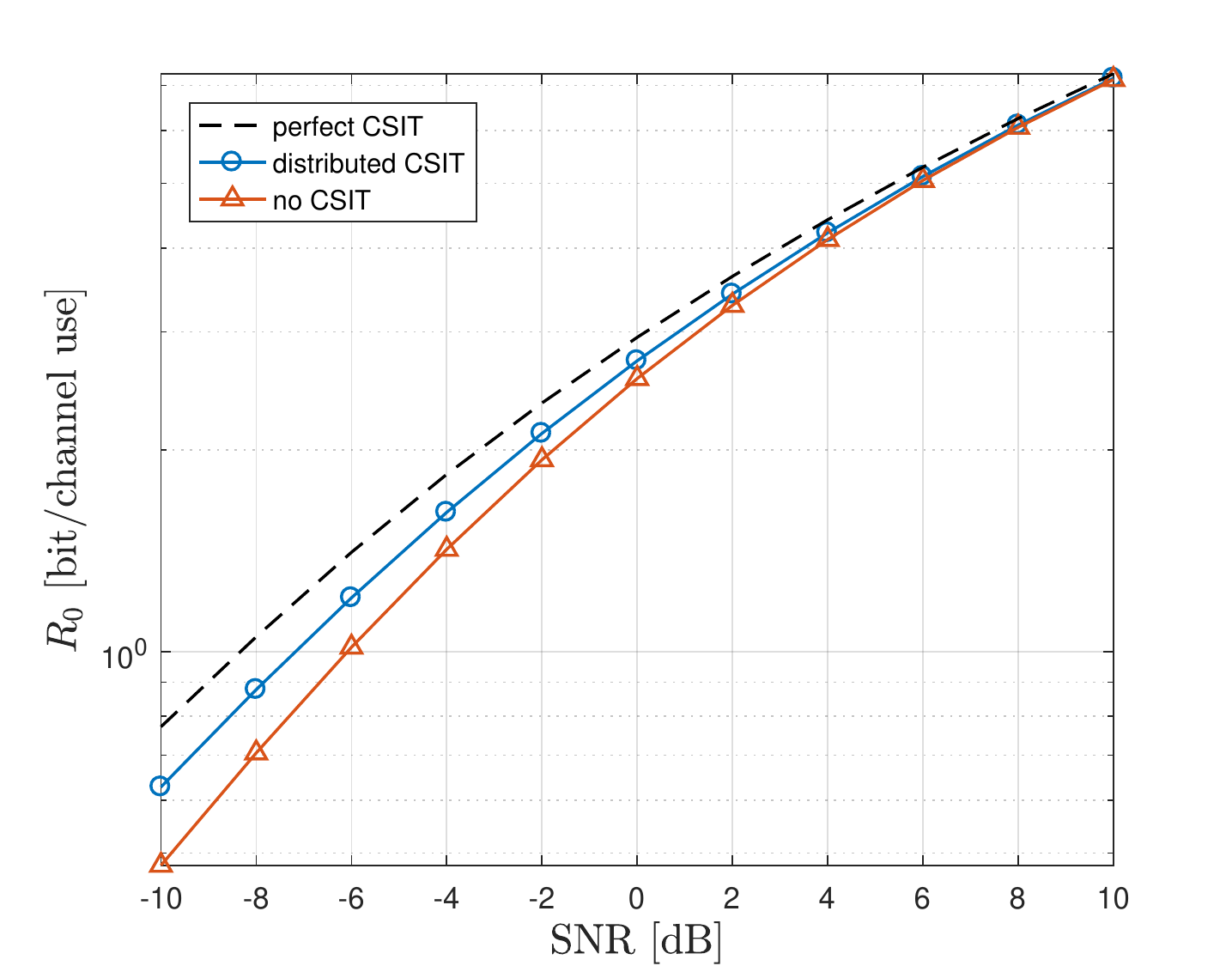}}
    \hfill \\
  \subfloat[$2\times 1$ Cooperative MISO]{
        \includegraphics[width=0.8\columnwidth]{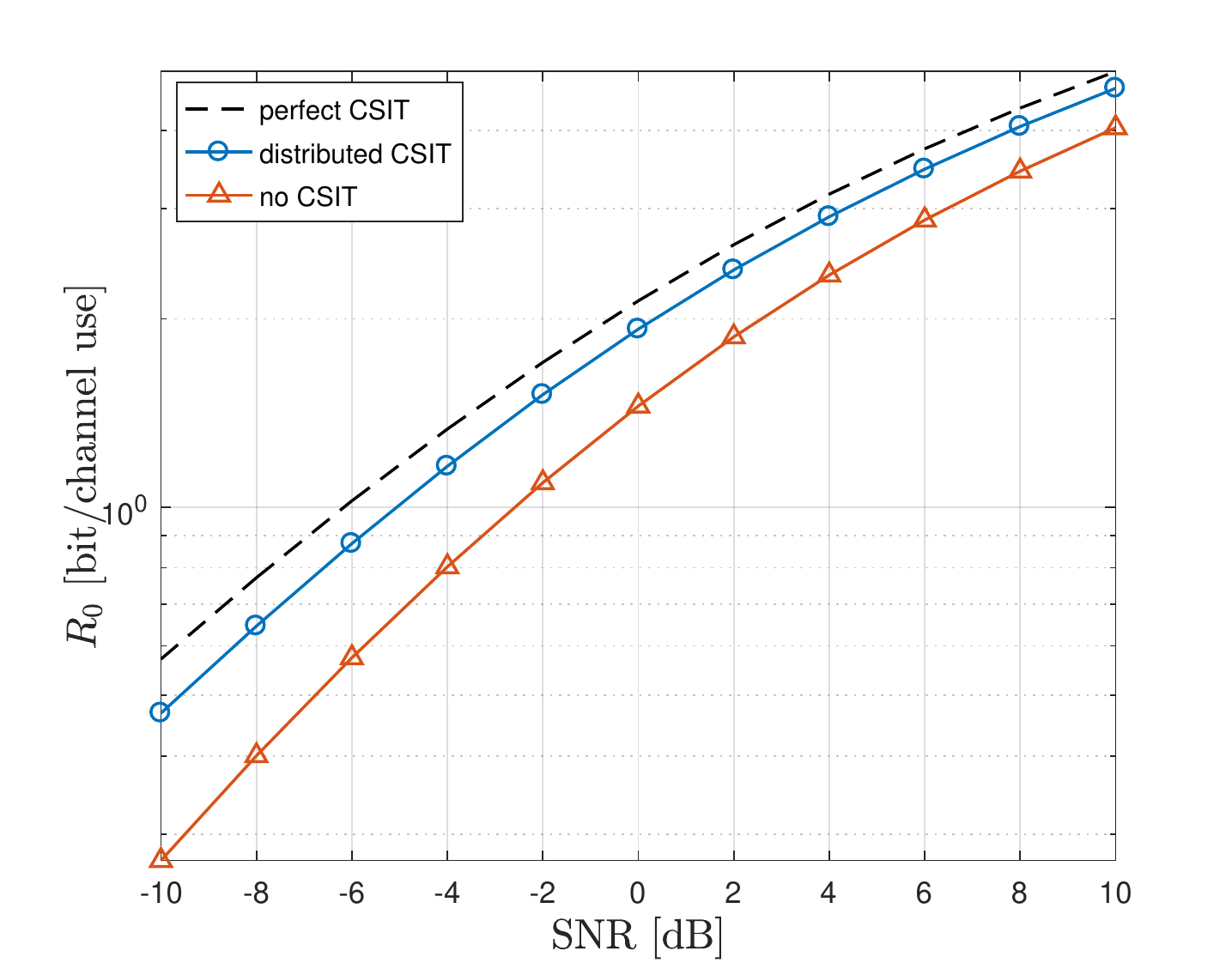}}
    \\
  \caption{Common message capacity \eqref{eq:C_distributed_MIMO} vs SNR for (a) 2 RX antennas and (b) single RX antenna.}
  \label{fig:C_sim} 
  %\vspace{-0.8cm}
\end{figure}

\section{Conclusion}
In this paper, we studied a two-user memoryless state-dependent multiple access channel under the assumption that
causal and distributed CSIT is available and messages can be partly or entirely shared prior to the data transmission. 
 We characterized the common message capacity of this channel and demonstrated that it is optimal to encode the message
as a function of current CSIT only based on Shannon strategies. We provide an insightful example over an additive binary-input quaternary-output channel with binary states showing that, interestingly, in some cases there is a threshold in terms of CSIT quality below which one encoder shall not use its channel knowledge. For a special case when CSIT is a deterministic function of CSIR, 
the full capacity region of a common message and two private messages is also characterized.  This last result is specialized to a practically relevant cooperative MIMO fading channel operating in FDD mode such that CSIT is acquired via an explicit feedback from the receiver. The cooperative MIMO example surprisingly reveals that in a distributed CSIT setup the optimal number of data streams shall not be restricted to the minimum number of transmit antennas. This is in contrast to the classical MIMO design under the centralized CSIT assumption. 

Interesting open problems include the evaluation of the minimum message cooperation (i.e., the minimum rate $R_0$) required such that the sum capacity is achievable via Shannon strategies, and the extension of the coding ideas derived for the cooperative MIMO case to systems with multiple receivers.  

\appendix
\section{Proofs - General Results}
\label{sec:proofs_general}
\subsection{Proof of Theorem \ref{th:C_common}}
\label{sec:proof_th1}
\begin{proof}[\unskip\nopunct]
\hspace{2\parindent}\textit{Converse:}
Let us define $U_i = (W_0,S_1^{i-1},S_2^{i-1})$. We construct an upper-bound by assuming that past CSIT realizations $(S_1^{i-1},S_2^{i-1})$ are available at both encoders. Hence, we assume that $X_{1i}$ and $X_{2i}$ are functions of $ (W_0,S_1^i,S_2^{i-1}) = (U_i,S_{1i})$ and $ (W_0,S_2^i,S_1^{i-1}) = (U_i,S_{2i})$ respectively. Note that  $U_i$ is independent of $(S_i,S_{1i},S_{2i},S_{Ri})$. Consider for brevity $\tilde{Y}_i=(Y_i,S_{Ri})$. We then have:
\begingroup
\allowdisplaybreaks
\begin{align}\label{eq:conv_rate}
\begin{split}
		nR_0  = & H(W_0) \\
							 = & I(W_0; \tilde{Y}^n) + H(W_0|\tilde{Y}^n) \\
							 \overset{(a)}{\leq} & I(W_0; \tilde{Y}^n) + n\epsilon_n \\
							 = & \sum_{i=1}^n I(W_0;\tilde{Y}_i|\tilde{Y}^{i-1}) + n\epsilon_n \\
							 = &\sum_{i=1}^n H(\tilde{Y}_i|\tilde{Y}^{i-1},) - H(\tilde{Y}_i|W_0,\tilde{Y}^{i-1}) + n\epsilon_n \\
							 \leq &\sum_{i=1}^n H(\tilde{Y}_i|\tilde{Y}^{i-1})  - H(\tilde{Y}_i|U_i,\tilde{Y}^{i-1}) + n\epsilon_n \\
							 \overset{(b)}{=} & \sum_{i=1}^n H(\tilde{Y}_i|\tilde{Y}^{i-1}) - H(\tilde{Y}_i|U_i) + n\epsilon_n \\
							 \leq &\sum_{i=1}^n H(\tilde{Y}_i) - H(\tilde{Y}_i|U_i) + n\epsilon_n \\
							 = &\sum_{i=1}^n I(U_i;\tilde{Y}_i) + n\epsilon_n\\
							 \overset{(c)}{=} &\sum_{i=1}^n I(U_i;Y_i|S_{Ri}) + n\epsilon_n\\
	 \end{split}
\end{align}
where $(a)$ follows from Fano's inequality ($\lim_{n\to \infty}\epsilon_n = 0$), $(b)$ follows from the Markov chain 
$\tilde{Y}^{i-1} \to (W_0,S_1^{i-1},S_2^{i-1}) \to \tilde{Y}_i$, and $(c)$ is because $S_{Ri}$ is independent of $U_i$.
The code must also satisfy the input cost constraints
\begin{equation}\label{eq:conv_cost}
    P_1 \geq \EE\left[\dfrac{1}{n}\sum_{i=1}^n b_1(X_{1i})\right], \quad P_2 \geq \EE\left[\dfrac{1}{n}\sum_{i=1}^n b_2(X_{2i})\right].
\end{equation}
We combine the bounds in \eqref{eq:conv_rate} and \eqref{eq:conv_cost} by means of a time-sharing variable $Q$ uniformly distributed in $\{1,\dots,n\}$ and independent of everything else, and by letting $U = (U_Q,Q)$, $X_1=X_{1Q}$, $X_2=X_{2Q}$ $Y=Y_Q$, $S=S_Q$, $S_1=S_{1Q}$, $S_2=S_{2Q}$, $S_{R}=S_{RQ}$. Note that the resulting distribution $p_{Y,X_1,X_2,S,S_1,S_2,S_R,U}$ factors as $
    p(y|x_1,x_2,s)\mathbbm{1}[x_1-f_1(u,s_1)]\mathbbm{1}[x_2-f_2(u,s_2)] p(s,s_1,s_2,s_R)p(u)$, where $\mathbbm{1}[\cdot]$ is an indicator function. With these identifications, we obtain
\begin{align*}
    R_0&\leq I(U;Y|S_R, Q) + \epsilon_n\\
    &\leq H(Y|S_R)-H(Y|U,S_R,Q)+ \epsilon_n\\
    &\overset{(c)}{=} H(Y|S_R) -H(Y|U,S_R)+ \epsilon_n\\
    &= I(U;Y|S_R) + \epsilon_n,
\end{align*}
where (c) follows from the Markov chain $Q\to (U,S_R) \to Y$, and
$P_1\geq \EE[b_1(X_1)]$, $P_2\geq \EE[b_2(X_2)]$. Hence, we finally have $R_0 \leq C_0(P_1,P_2) + \epsilon_n$.
\end{proof}
\endgroup

\begin{proof}[\unskip\nopunct]
\hspace{2\parindent}\textit{Achievability:} Achievability follows from standard arguments, hence a formal proof is omitted. An informal yet intuitive proof can be obtained by the classical \textit{physical device} argument of Shannon \cite{shannon1958channels,el2011network}. More precisely, by fixing the functions $f_k(u_k,s_k)$, we can consider a new state-less and memoryless MAC with messages $W_0,W_1,W_2$, inputs $U_1,U_2$, and output $(Y,S_R)$. For a given $f_k(u_k,s_k)$, the capacity region of this auxiliary channel is simply achievable by Slepian-Wolf coding for the MAC with common and independent messages \cite{slepian1973coding}, which gives the achievable region in Lemma \ref{lem:MAC}. 

The finite cardinality of $U_1,U_2$ follows directly by Shannon argument, which states that using Shannon strategies corresponds to coding over an augmented input alphabet of functions $\set{S}_k \to \set{X}_k$ of size $|\set{X}_k|^{|\set{S}_k|}$, indexed by $U_k$ \cite{el2011network}. Hence we can consider $|\set{U}_k|\leq |\set{X}_k|^{|\set{S}_k|}$. The finite cardinality of $\set{U}_0$ follows by a simple application of the support lemma \cite[Appendix~C]{el2011network} applied to the Slepian-Wolf region of the auxiliary channel, which gives \cite[p.~344]{el2011network}
\begin{equation*}
    |\set{U}_0| \leq \min\{|\set{X}_1|^{|\set{S}_1|}|\set{X}_2|^{|\set{S}_2|}+2, |\set{Y}||\set{S}_R|+3 \}.
\end{equation*}
Finally, the expression in Theorem \ref{th:C_common} can be obtained by specializing the proof of Lemma \ref{lem:MAC} to the transmission of a common message only, i.e. by letting $R_1=R_2=0$, and by identifying $U_0 = U_1 = U_2 \eqdef U$.
\end{proof}

\subsection{Proof of Theorem \ref{th:MAC_det}}
\label{sec:proof_th2}
\begin{proof}[\unskip\nopunct]
\hspace{2\parindent}\textit{Achievability:}
The proof follows the same lines as in \cite{jafar2006capacity}. Achievability builds on Lemma \ref{lem:MAC}, where we rewrite the mutual information terms as follows. By focusing first on the sum-rate, we observe that
\begin{align*}
    I&(U_1,U_2;Y|S_R) =\\
    &=H(Y|S_R)-H(Y|U_1,U_2,S_R) \\
               &\overset{(a)}{=} H(Y|S_R)-H(Y|U_1,U_2,S_1,S_2,S_R) \\ 
               &\overset{(b)}{=} 
               H(Y|S_R)-H(Y|X_1,X_2,U_1,U_2,S_1,S_2,S_R) \\
               &\overset{(c)}{=} 
               H(Y|S_R)-H(Y|X_1,X_2,S_R) \\
               &= I(X_1,X_2;Y|S_R),
\end{align*}
where $(a)$ comes from $(S_1,S_2) = (q_1(S_R),q_2(S_R))$, $(b)$ is because $(X_1,X_2)$ is a function of $(S_1,S_2,U_1,U_2)$, and $(c)$ because of the Markov chain $(S_1,S_2,U_1,U_2)\to(X_1,X_2,S_R)\to Y$. Similarly, one can show
\begin{align*}
    I&(U_1;Y|S_R,U_2,U_0) = H(Y|X_2,U_2,U_0,S_2,S_R)\\
    &\quad -H(Y|X_1,X_2,U_1,U_2,U_0,S_1,S_2,S_R) \\ 
               &\overset{(a)}{=} H(Y|X_2,U_0,S_R)-H(Y|X_1,X_2,U_0,S_R) \\
               &= I(X_1;Y|X_2,U_0,S_R),
\end{align*}
where $(a)$ follows from $(S_2,U_2)\to(U_0,X_2,S_R)\to Y$ and $(S_1,S_2,U_1,U_2)\to(U_0,X_1,X_2,S_R)\to Y$, and
\begin{align*}
    I(U_2;Y|S_R,U_1,U_0) &= I(X_2;Y|X_1,U_0,S_R),\\
    I(U_1,U_2;Y|S_R,U_0) &= I(X_1,X_2;Y|U_0,S_R).
\end{align*}

Finally, by the functional representation lemma \cite[Appendix~B]{el2011network} and since $U_1$ and $U_2$ do not appear in the bounds given by \eqref{eq:MAC_det} and proven above, designing 
\begin{equation*}
    p(u_0)p(u_1|u_0)p(u_2|u_0)\mathbbm{1}[x_1-f_1(s_1,u_1)]\mathbbm{1}[x_2-f_2(s_2,u_2)],
\end{equation*}
where $\mathbbm{1}[\cdot]$ is an indicator function, is equivalent to designing $p(u)p(x_1|s_1,u)p(x_2|s_2,u)$ (replacing $U_0$ with $U$).
\end{proof}

\begin{proof}[\unskip\nopunct]
\hspace{2\parindent}\textit{Converse:}
Let us define $U_{i} = (W_0,S_1^{i-1},S_2^{i-1})$. We construct an outer bound by assuming that past CSIT realizations $(S_1^{i-1},S_2^{i-1})$ are available at both encoders. Hence, we assume that $X_{1i}$ and $X_{2i}$ are functions of $ (W_0,W_1,S_1^i,S_2^{i-1}) = (W_1,U_{i},S_{1i})$ and $ (W_0,W_2,S_2^i,S_1^{i-1}) = (W_2,U_{i},S_{2i})$ respectively. Note that  $U_{i}$ is independent of $(S_i,S_{1i},S_{2i},S_{Ri})$. We then bound
\begingroup
\allowdisplaybreaks
\begin{align*}
		nR_1 \overset{(a)}{\leq} & I(W_1; Y^n,S_R^n|W_0,W_2) + n\epsilon_n \\
			  \overset{(b)}{=} & I(W_1; Y^n|W_0,W_2,S_R^n) + n\epsilon_n \\
			 = & \sum_{i=1}^n I(W_1;Y_i|Y^{i-1},W_0, W_2,S_R^n) + n\epsilon_n \\
			 = &\sum_{i=1}^n H(Y_i|W_0,W_2,Y^{i-1},S_R^n)\\
			 & - H(Y_i|W_0,W_1,W_2,Y^{i-1},S_R^n) + n\epsilon_n \\
			 \overset{(c)}{=} &\sum_{i=1}^n H(Y_i|W_0,W_2,S_1^{i-1},S_2^i,Y^{i-1},S_R^n)\\
			 & - H(Y_i|W_0,W_1,W_2,S_1^i,S_2^i,Y^{i-1},S_R^n) + n\epsilon_n \\
			 \overset{(d)}{=} & \sum_{i=1}^n H(Y_i|W_2,X_{2i},U_{i},S_{2i},Y^{i-1},S_R^n) \\
			 &- H(Y_i|X_{1i},X_{2i},U_{i},S_{Ri}) + n\epsilon_n \\
			 \leq &\sum_{i=1}^n H(Y_i|X_{2i},U_{i},S_{Ri}) \\
			 &- H(Y_i|X_{1i},X_{2i},U_{i},S_{Ri}) + n\epsilon_n \\
			 = &\sum_{i=1}^n I(X_{1i};Y_i|X_{2i},U_{i},S_{Ri}) + n\epsilon_n,
\end{align*}
\endgroup
where (a) follows from Fano's inequality ($\lim_{n\to \infty}\epsilon_n = 0$), (b) from the independence of $W_1$ and $S_R^n$, $(c)$ from $(S_{i1},S_{2i})$ being a function of $S_{Ri}$, and (d) from the Markov chain $(W_0,W_1,W_2,S_1^i,S_2^i,Y^{i-1},\{S_{Rj}\}_{j\neq i})\to (X_{1i},X_{2i},S_{Ri}) \to Y_i$. Similarly, we have 
\begin{align*}
        nR_2 \leq & \sum_{i=1}^n I(X_{2i};Y_i|X_{1i},U_{i},S_{Ri}) + n\epsilon_n,\\
		n(R_1+R_2) \leq &\sum_{i=1}^n I(X_{1i},X_{2i};Y_i|U_{i},S_{Ri}) + n\epsilon_n,\\
		nR_{\mathrm{sum}} \leq &\sum_{i=1}^n I(X_{1i},X_{2i};Y_i|S_{Ri}) + n\epsilon_n.
\end{align*}
The code must also satisfy the input cost constraints
\begin{equation*}
   P_1 \geq \EE\left[\dfrac{1}{n}\sum_{i=1}^n b_1(X_{1i})\right], \quad P_2 \geq \EE\left[\dfrac{1}{n}\sum_{i=1}^n b_2(X_{2i})\right].
\end{equation*}
We combine all the bounds by means of a time-sharing variable $Q$ uniformly distributed in $\{1,\dots,n\}$ and independent of everything else, and by letting $U = (U_{Q},Q)$, $X_1=X_{1Q}$, $X_2=X_{2Q}$ $Y=Y_Q$, $S=S_Q$, $S_1=S_{1Q}$, $S_2=S_{2Q}$, $S_{R}=S_{RQ}$. Note that the resulting distribution $p_{Y,X_1,X_2,S,S_1,S_2,S_R,U}$ factors as
\begin{equation*}
    p(y|x_1,x_2,s)p(x_1|s_1,u)p(x_2|s_2,u)p(s,s_1,s_2,s_R)p(u)
\end{equation*} as required. With these identifications, we readily obtain
\begin{equation*}
\begin{gathered}
    R_1 \leq I(X_1;Y|X_2,U,S_R) + \epsilon_n \\
    R_2 \leq I(X_2;Y|X_1,U,S_R) + \epsilon_n \\
    R_1+R_2 \leq I(X_1,X_2;Y|U,S_R) + \epsilon_n \\
    \begin{aligned}
    R_{\mathrm{sum}}&\leq I(X_1,X_2;Y|S_R, Q) + \epsilon_n\\
    & \leq I(X_1,X_2;Y|S_R) + \epsilon_n, 
    \end{aligned}\\
    P_1 \geq \EE[b_1(X_1)], P_2 \geq \EE[b_2(X_2)]. 
\end{gathered}
\end{equation*}
\end{proof}

\subsection{Proof of Proposition \ref{prop:outer_cond_indep}}
\label{ssec:proof_outer_cond_indep}
Let us define $U_{0i} :=(W_0,S_R^{i-1})$, $U_{1i}:= (W_1,S_1^{i-1},U_{0i})$, and $U_{2i}:= (W_2,S_2^{i-1},U_{0i})$. Note that $X_{1i}$ and $X_{2i}$ are functions of $(U_{1i},S_{1i})$ and $(U_{2i},S_{2i})$ respectively, and, due to the Markov chain $S_{1i} \to S_{Ri} \to S_{2i}$, we also have $U_{1i}\to U_{0i} \to U_{2i}$ as required. By Fano's inequality ($\lim_{n\to \infty}\epsilon_n = 0$), and by following similar steps as in the previous sections, we obtain
\begin{align*}
&n(R_1+R_2)\\
&\leq I(W_1,W_2;Y^n,S_R^n|W_0) + n\epsilon_n \\
		   &= I(W_1,W_2;Y^n|W_0,S_R^n) + n\epsilon_n \\
		   &= \sum_{i=1}^nI(W_1,W_2;Y_i|Y^{i-1},W_0,S_R^n) +n\epsilon_n \\	
		   &\leq \sum_{i=1}^nI(W_1,W_2,S_1^{i-1},S_2^{i-1};Y_i|Y^{i-1},W_0,S_R^n) +n\epsilon_n 	\\
& = \sum_{i=1}^n I(U_{1i},U_{2i};Y_i|Y^{i-1},U_{0i},S_{Ri}^n) +n\epsilon_n   \\
& \leq \sum_{i=1}^n I(U_{1i},U_{2i};Y_i|U_{0i},S_{Ri}) +n\epsilon_n,
\end{align*}
where the last inequality comes from the memoryless property of the channel. Following similar lines one can prove 
\begin{equation*}
nR_{\mathrm{sum}} \leq \sum_{i=1}^nI(U_{1i},U_{2i};Y_i|S_{Ri}) +n\epsilon_n,
\end{equation*}
which can be combined with the bound on $R_1+R_2$ and the power constraints by means of the usual time-sharing step.

\section{Proofs - FDD Cooperative MIMO Channel with Fading}\label{sec:proofs_MIMO}

\subsection{Proof of Theorem \ref{th:C_region_MIMO}}
\label{ssec:proof_capacity_region}
We construct an outer bound $\mathscr{C}_o(P_1,P_2)$ by following similar steps as in \cite{liu2006gausscommon}, but starting from the single-letter formulation of $\mathscr{C}(P_1,P_2)$ given by Theorem \ref{th:MAC_det} extended to continuous alphabets similarly to \cite{el2011network,caire1999capacity,caire2007adaptive}.
We consider the following applications of the \textit{maximum differential entropy lemma} \cite[p. 21]{el2011network}, adapted to the complex field. From the rightmost bound in \cite[Eq. 2.6]{el2011network}, we have
\begin{align*}
h(\rvec{Y}|\rmat{S}=\dmat{S}) &\leq \log\left((\pi e)^2\mathrm{det}\left(\EE\left[\rvec{Y}\rvec{Y}^\herm | \rmat{S}=\dmat{S} \right]\right)\right)\\
&= \log\left((\pi e)^2\mathrm{det}\left(\dmat{S}\dmat{\Sigma}\left( q_1(\dmat{S}), q_2(\dmat{S})\right)\dmat{S}^\herm + \dmat{I}\right)\right),
\end{align*}
where $\dmat{\Sigma}(S_1,S_2) = \EE\left[\rvec{X}\rvec{X}^\herm | S_1 , S_2\right]$. Then, from the bound in \cite[Eq. 2.7]{el2011network}, we have
\begin{align*}
h(\rvec{Y}&|U,\rmat{S}=\dmat{S}) \leq  \log\Big((\pi e)^2\mathrm{det}\Big(\EE\Big[(\rvec{Y}-\EE[\rvec{Y}|U,\rmat{S}=\dmat{S}])\\
&\quad \times (\rvec{Y}-\EE[\rvec{Y}|U,\rmat{S}=\dmat{S}])^\herm |\rmat{S}=\dmat{S} \Big]\Big)\Big)\\
&= \log\left((\pi e)^2\mathrm{det}\left(\dmat{S}\dmat{\Gamma}\left( q_1(\dmat{S}), q_2(\dmat{S})\right)\dmat{S}^\herm + \dmat{I}\right)\right),
\end{align*}
where $\dmat{\Gamma}(S_1,S_2) =$
\begin{equation*}
\dmat{\Sigma}(S_1,S_2)-\EE\left[\EE[\rvec{X}|U,S_1,S_2]\EE[\rvec{X}^\herm|U,S_1,S_2] | S_1, S_2 \right].
\end{equation*} 
By the structure of the input distribution, we now observe that $\dmat{\Gamma}(S_1,S_2) = \mathrm{diag}(\gamma_1(S_1),\gamma_2(S_2))$ and that, $\forall (s_1,s_2)\in\set{S}_1\times\set{S}_2$, 
\begin{equation}\label{eq:input_covariance}
\begin{split}
&\dmat{\Sigma}(s_1,s_2) =  \begin{bmatrix}
\gamma_1(s_1) & 0 \\ 0 & \gamma_2(s_2)
\end{bmatrix} \\
&\quad + \begin{bmatrix}
\EE[|\mu_1(U,s_1)|^2] & \EE[\mu_1(U,s_1)\mu_2^\star(U,s_2)] \\ \EE[\mu_2(U,s_2)\mu_1^\star(U,s_1) ] & \EE[|\mu_2(U,s_2)|^2]
\end{bmatrix},
\end{split}
\end{equation}
where we define the functions
\begin{align*}
\mu_k(U,S_k) &\eqdef \EE[X_k|U,S_k]\\
\gamma_k(S_k) &\eqdef \EE\left[|X_k|^2|S_k\right] - \EE\left[|\mu_k(U,S_k)|^2|S_k\right] \geq 0.
\end{align*}
By following similar steps for $h(\rvec{Y}|X_1,U,\rmat{S}=\dmat{S})$ and $h(\rvec{Y}|X_2,U,\rmat{S}=\dmat{S})$, and by applying the resulting bounds to the mutual information terms in Theorem \ref{th:MAC_det}, we obtain 
\begin{equation}\label{eq:Co}
\begin{gathered}
R_1 \leq \EE\left[\log\left(1+\gamma_1(S_1)\|\rmat{S}\dvec{e}_1\|^2 \right)\right],\\
R_2 \leq \EE\left[\log\left(1+\gamma_2(S_2)\|\rmat{S}\dvec{e}_2\|^2 \right)\right], \\
R_1+R_2 \leq \EE\left[\log\mathrm{det}\left(\dmat{I}+\rmat{S}\mathrm{diag}(\gamma_1(S_1),\gamma_2(S_2))\rmat{S}^\herm\right)\right],\\
R_0+R_1+R_2 \leq \EE\left[\log\mathrm{det}\left(\dmat{I}+\rmat{S}\dmat{\Sigma}(S_1,S_2)\rmat{S}^\herm\right)\right].
\end{gathered}
\end{equation}
The outer bound $\mathscr{C}_{o}(P_1,P_2)$ is then established by taking the convex hull of the union of all rate triples $(R_0,R_1,R_2)$ satisfying \eqref{eq:Co}
for some $p(x_1|s_1,u)p(x_2|s_2,u)p(u)$ such that $\EE[|X_k|^2]\leq P_k$.

Similarly to the proof of Theorem \ref{th:C_distributed_MIMO}, it can be now shown that every $\dmat{\Sigma}(S_1,S_2)$ as in \eqref{eq:input_covariance} induced by any $p(x_1|s_1,u)p(x_2|s_2,u)p(u)$ can also be obtained by the scheme in \eqref{eq:achievable_scheme_MAC}, i.e., via superposition of linearly precoded Gaussian codes. The key point is showing that the second term in the RHS of \eqref{eq:input_covariance} can be obtained via distributed linear precoders of dimension $d$. This follows by the same technique used in the proof of Theorem \ref{th:C_distributed_MIMO}, by simply replacing the functions $f_k(u,s_k)$ with $\mu_k(u,s_k)$. Finally, since the inputs are conditionally Gaussian, standard arguments \cite{el2011network} show that \eqref{eq:achievable_scheme_MAC} attains the bound $\mathscr{C}_o(P_1,P_2)$, without time sharing (i.e., we can omit the convex hull operation).

\subsection{Proof of Proposition \ref{cor:necessary}}\label{ssec:necessary_proof}
The proof is split for the sake of clarity in the following three steps:
\begin{enumerate}[leftmargin=*]
    \item We fix a specific conditional input covariance matrix $\dmat{\Sigma}^\star(S_1,S_2)$, and we show that it is achievable via distributed linear precoding if and only if $d'>2$.
    \item We construct a specific $p(\dmat{S},s_1,s_2)$ such that $\dmat{\Sigma}^\star(S_1,S_2)$ is the unique optimal solution to Problem \eqref{eq:C_distributed_MIMO}.
    \item We combine the above steps to show that there exists a channel for which $d' \leq 2$ leads to strictly suboptimal rates.
\end{enumerate}

\textit{Step 1}: Consider binary D-CSIT alphabets, i.e., $\set{S}_1 = 
\set{S}_2 = \{0,1\}$, and let $\dmat{\Sigma}^\star(S_1,S_2)$ be given by
\begin{align}\label{eq:Sigma_star}
\begin{split}
    &\dmat{\Sigma}^\star(0,0) = \dmat{I}, \quad \dmat{\Sigma}^\star(1,0) = \dmat{I}, \\
    &\dmat{\Sigma}^\star(0,1) = \bigl[\begin{smallmatrix}1 & 0.6\\0.6 & 1
  \end{smallmatrix}\bigr], \quad \dmat{\Sigma}^\star(1,1) = \bigl[\begin{smallmatrix}1 & 0.8\\0.8 & 1
  \end{smallmatrix}\bigr].
  \end{split}
\end{align} 
Define the set $\set{G}(\tilde{d})$ of conditional input covariance matrices $\dmat{\Sigma}(S_1,S_2)$ which are achievable via distributed linear precoders of maximal dimension $\tilde{d}$, i.e., 
\begin{align}\label{eq:implementable}
    \set{G}(\tilde{d})&\eqdef 
  \begin{cases}
    \dmat{\Sigma}(S_1,S_2) \in \stdset{S}^2_+ \quad \text{s.t.}\\
    \dmat{\Sigma}(S_1,S_2) = \begin{bmatrix}
    \dvec{g}_1^\herm(S_1) \\ \dvec{g}_2^\herm(S_2) 
\end{bmatrix} \begin{bmatrix}
    \dvec{g}_1(S_1) & \dvec{g}_2(S_2)
\end{bmatrix},\\
    \dvec{g}_1(S_1), \: \dvec{g}_2(S_2) \in \stdset{C}^{d'}, \;  d'\leq \tilde{d}
  \end{cases}
\end{align}
Clearly, $\set{G}(d_1) \subseteq \set{G}(d_2)$, for $d_1 \leq d_2$. The following lemma holds:
\begin{lemma}\label{lem:Sigma_star}
$\dmat{\Sigma}^\star(S_1,S_2) \in \set{G}(3)$, and $\dmat{\Sigma}^\star(S_1,S_2) \notin \set{G}(2)$.
\end{lemma}
\begin{IEEEproof}
For $\dmat{\Sigma}^\star(S_1,S_2)$ to be achievable, we need to find precoders $\dvec{g}_k(S_k)$ s.t.
\begin{equation*}
    \begin{cases}
    \dvec{g}_1^\herm(0)\dvec{g}_2(0) = 0, \quad \dvec{g}_1^\herm(0)\dvec{g}_2(1) = 0.6, \\ \dvec{g}_1^\herm(1)\dvec{g}_2(0) = 0, \quad \dvec{g}_1^\herm(1)\dvec{g}_2(1) = 0.8, \\
    \|\dvec{g}_1(0)\| = \|\dvec{g}_1(1)\| = \|\dvec{g}_2(0)\| = \|\dvec{g}_2(1)\| = 1.
    \end{cases}
\end{equation*}
For $\dvec{g}_k(S_k)$ of dimension $d'=2$, the above system has no solution. In fact, we need to simultaneously satisfy
\begin{equation*}
    \begin{cases}
    \dvec{g}_1(0) \perp \dvec{g}_2(0), \\
    \dvec{g}_1(1) \perp \dvec{g}_2(0), \\
    \|\dvec{g}_1(0)\| = \|\dvec{g}_1(1)\| =  1,
    \end{cases}
\end{equation*}
which, for $d'=2$, implies $\dvec{g}_1(0) = \pm \dvec{g}_1(1)$, and hence leads to the following contradiction $
        0.6 = \dvec{g}_1^\herm(0)\dvec{g}_2(1) = \pm \dvec{g}_1^\herm(1)\dvec{g}_2(1) = \pm 0.8$.
Instead, one can check that $\dmat{\Sigma}^\star(S_1,S_2)$ is readily obtained by letting $d'=3$ and
\begin{align*}
   \dvec{g}_1(0) &= \begin{bmatrix}1 & 0 & 0
  \end{bmatrix}, \quad 
   \dvec{g}_1(1) = \begin{bmatrix}0 & 1 & 0
  \end{bmatrix},\\
  \dvec{g}_2(0) &= \begin{bmatrix}0 & 0 & 1
  \end{bmatrix}, \quad 
   \dvec{g}_2(1) = \begin{bmatrix}0.6 & 0.8 & 0
  \end{bmatrix}.
\end{align*}
\end{IEEEproof}

\textit{Step 2:} Consider the following rewriting of Problem \eqref{eq:C_distributed_MIMO}, by letting again  $\set{S}_1 = 
\set{S}_2 = \{0,1\}$ (hence $d = 4$), and unitary power constraint $P_1 = P_2 = 1$:
\begin{equation}
    C_0 = \max_{\dmat{\Sigma}\in \set{P}\cap \set{G}(4)}\EE\left[\log\mathrm{det}\left(\dmat{I}+\rmat{S}\dmat{\Sigma}(S_1,S_2)\rmat{S}^\herm\right)\right],\label{eq:rewriting_MIMO}
\end{equation}
where $\set{G}(4)$ is given by \eqref{eq:implementable}, and where
\begin{equation*}
    \set{P}\eqdef \{ \dmat{\Sigma}(S_1,S_2) \in \stdset{S}^2_+ \;|\; \EE[\Sigma_{k,k}(S_1,S_2)] \leq 1, \: k =1,2 \}
\end{equation*}
is the per-TX power constraint. Note that $\dmat{\Sigma}^\star(S_1,S_2)$ belongs to the feasible set, i.e.  $\dmat{\Sigma}^\star(S_1,S_2) \in \set{P}\cap \set{G}(4)$.
\begin{lemma}\label{lem:optimal_unique}
There exist some $p(\dmat{S},s_1,s_2)$ such that $\dmat{\Sigma}^\star(S_1,S_2)$ given by \eqref{eq:Sigma_star} is the unique optimal solution for problem \eqref{eq:rewriting_MIMO}.
\end{lemma}
\begin{IEEEproof}
The main idea is to build such CSI distribution by ``reversing'' a spatio-temporal water-filling algorithm which gives as unique optimal solution the conditional input covariance $\dmat{\Sigma}^\star(S_1,S_2)$. We now provide the details.

Define a uniformly distributed random state $\rmat{S}$ taking values in the finite alphabet  $\set{S} =  \{\dmat{S}_1,\dmat{S}_2,\dmat{S}_3,\dmat{S}_4\}$, and let the CSIT be given by the functions
\begin{align*}
    s_1 = q_1(\dmat{S}) &= \begin{cases}
0 \text{ for } \dmat{S} \in \{\dmat{S}_1,\dmat{S}_2\} \\
1 \text{ otherwise } \\
\end{cases} \\
s_2 = q_2(\dmat{S}) &= \begin{cases}
0 \text{ for } \dmat{S} \in \{\dmat{S}_1,\dmat{S}_3\} \\
1 \text{ otherwise } \\
\end{cases} 
\end{align*}
The capacity of such a channel can be upper bounded by
\begin{align}
    C_0 &= \max_{\dmat{\Sigma}\in \set{P}\cap \set{G}(4)}\EE\left[\log\mathrm{det}\left(\dmat{I}+\rmat{S}\dmat{\Sigma}(S_1,S_2)\rmat{S}^\herm\right)\right],\\
     &\leq \max_{\dmat{\Sigma}\in \set{P}}\EE\left[\log\mathrm{det}\left(\dmat{I}+\rmat{S}\dmat{\Sigma}(S_1,S_2)\rmat{S}^\herm\right)\right],\label{eq:half_relaxed_MIMO}\\
    & \leq \max_{\dmat{\Sigma}\in \set{P}'}\EE\left[\log\mathrm{det}\left(\dmat{I}+\rmat{S}\dmat{\Sigma}(S_1,S_2)\rmat{S}^\herm\right)\right], \label{eq:relaxed_MIMO}
\end{align}
where
\begin{equation*}
    \set{P}'\eqdef \{ \dmat{\Sigma}(S_1,S_2) \in \stdset{S}^2_+ \;|\; \mathrm{tr}\{\EE[\dmat{\Sigma}(S_1,S_2)]\} \leq 2 \},
\end{equation*}
is the set obtained by relaxing the per-TX power constraint $\set{P}$ to a total power constraint ($\set{P}\subseteq \set{P}')$. Inequalities  \eqref{eq:half_relaxed_MIMO} and \eqref{eq:relaxed_MIMO} are obtained respectively by relaxing the achievability via distributed linear precoding and the power constraint.

Problem \eqref{eq:relaxed_MIMO} turns out to be an instance of a classical (centralized) MIMO capacity problem, where the optimal solution is given by the well-known spatio-temporal water-filling algorithm. More precisely, let us rewrite \eqref{eq:relaxed_MIMO} as
\begin{equation}
\begin{split}
    &\max_{\dmat{\Sigma}(S_1,S_2)\in \set{P}'}\sum_{i=1}^4p(\dmat{S}_i)\log\mathrm{det}\left(\dmat{I}+\dmat{S}_i\dmat{\Sigma}(q_1(\dmat{S}_i),q_2(\dmat{S}_i))\dmat{S}_i^\herm\right)\\
    &= \max_{\{\dmat{\Sigma}_i \} \in \tilde{\set{P}}'}\dfrac{1}{4}\sum_{i=1}^4\log\mathrm{det}\left(\dmat{I}+\dmat{S}_i\dmat{\Sigma}_i\dmat{S}^\herm\right),\end{split}\label{eq:relaxed_MIMO_2}
\end{equation}
where we defined $\dmat{\Sigma}_i \eqdef \dmat{\Sigma}(q_1(\dmat{S}_i),q_2(\dmat{S}_i)) $, and where 
\begin{equation*}
    \tilde{\set{P}}' = \Bigl\{ \dmat{\Sigma}_i \in \stdset{S}^2_+  \; | \; \frac{1}{4}\sum_{i=1}^4 \mathrm{tr}\{\dmat{\Sigma}_i\} \leq 2 \Bigr\}.
\end{equation*}
A well-known application of the Hadamard's inequality gives the following upper bound in terms of the channel eigen-decompositions $\dmat{S}_i^\herm\dmat{S}_i = \dmat{V}_i\dmat{\Lambda}_i \dmat{V}_i^\herm$, $\dmat{\Lambda}_i = \mathrm{diag}(\lambda_{i,1},\lambda_{i,2})$ 
\begin{equation*}
\begin{split}
    &\max_{\{\dmat{\Sigma}_i \} \in \tilde{\set{P}}'}\dfrac{1}{4}\sum_{i=1}^4\log\mathrm{det}\left(\dmat{I}+\dmat{S}_i\dmat{\Sigma}_i\dmat{S}_i^\herm\right) \\
    &\leq \max_{\substack{\xi_{i,k}\geq 0\\ \frac{1}{4}\sum_{i,k}\xi_{i,k}\leq 2}}\dfrac{1}{4}\sum_{i=1}^4\sum_{k=1}^2\log(1+\lambda_{i,k}\xi_{i,k}),
    \end{split}
\end{equation*}
where the optimal $\xi_{i,k}$ are given by the water-filling conditions
\begin{equation*}
    \xi_{i,k} = \max\Bigl\{\nu - \dfrac{1}{\lambda_{i,k}},0 \Bigr\}, \quad i = 1,\ldots,4, \quad k=1,2,
\end{equation*}
\begin{equation*}
    \sum_{i,k}\max\Bigl\{\nu - \dfrac{1}{\lambda_{i,k}},0 \Bigr\} = 8,
\end{equation*}
and where equality is achieved for
\begin{equation*}
    \dmat{\Sigma}_i=\dmat{V}_i\dmat{\Xi}_i \dmat{V}_i^\herm, \quad \dmat{\Xi}_i = \mathrm{diag}(\xi_{i,1},\xi_{i,2}),\quad i=1,\ldots,4.
\end{equation*}

Consider the conditional covariance $\dmat{\Sigma}^\star(S_1,S_2)$ given by \eqref{eq:Sigma_star}. Note that $\dmat{\Sigma}^\star(S_1,S_2) \in \set{P}'$, i.e., it satisfies the total power constraint. We wish to construct $\set{S} = \{ \dmat{S}_i \}$ such that $\dmat{\Sigma}^\star(S_1,S_2)$ is the unique optimal solution for \eqref{eq:relaxed_MIMO}. This can be done by ``reversing'' the MIMO water-filling algorithm described above. More precisely, let us consider $\dmat{\Sigma}^\star_i \eqdef \dmat{\Sigma}^\star(q_1(\dmat{S}_i),q_2(\dmat{S}_i)) $ and their eigen-decompositions 
\begin{equation*}
    \dmat{\Sigma}^\star_i=\dmat{V}^\star_i\dmat{\Xi}_i^\star \dmat{V}^{\star\herm}_i, \quad \dmat{\Xi}_i^\star = \mathrm{diag}(\xi^\star_{i,1},\xi^\star_{i,2}).
\end{equation*}
We construct now $\set{S}$ by letting
\begin{equation*}
    \dmat{S}_i = (\dmat{V}^\star_i\dmat{\Lambda}^\star_i\dmat{V}^{\star\herm}_i)^{\frac{1}{2}},\quad i=1,\ldots,4 
\end{equation*}
where the eigenvalues $\dmat{\Lambda}^\star_i = \mathrm{diag}(\lambda^\star_{i,1},\lambda^\star_{i,2})$ are given by
\begin{equation*}
    \lambda_{i,k}^\star = \frac{1}{\nu^\star - \xi_{i,k}^\star}, \quad i = 1,\ldots,4, \quad k=1,2,
\end{equation*}
and any choice of $\infty > \nu^\star > \max_{i,k}\xi_{i,k}^\star = 1.8$. 

By construction, $\dmat{\Sigma}^\star(S_1,S_2)$ is an optimal solution for \eqref{eq:relaxed_MIMO}. Uniqueness of the solution can be proven by contradiction as in  \cite[Section~III.A]{boyd2004waterfilling}, or directly by the strict concavity of $\sum_{i=1}^4\log\mathrm{det}\left(\dmat{I}+\dmat{S}_i\dmat{\Sigma}_i\dmat{S}_i^\herm\right)$ in $ \{\dmat{\Sigma}_i \succeq 0\}$, which is a direct consequence of the strict concavity of $\log \mathrm{det}(\bm{A})$ in $\bm{A} \succ \vec{0}$ and of the positive definiteness of $\dmat{S}_i$ by construction. 
Finally, since $\dmat{\Sigma}^\star(S_1,S_2) \in \set{P}\cap \set{G}(4)$, \eqref{eq:relaxed_MIMO} and \eqref{eq:half_relaxed_MIMO} are satisfied with equality.
\end{IEEEproof}

\textit{Step 3:}
The proof is now concluded by combining Lemma \ref{lem:Sigma_star} and Lemma \ref{lem:optimal_unique}, yielding $\exists \: p(\dmat{S},s_1,s_2)$ such that
\begin{equation*}
\begin{split}
    &\arg \max_{\dmat{\Sigma}(S_1,S_2)\in \set{P}\cap \set{G}(4)}\EE\left[\log\mathrm{det}\left(\dmat{I}+\rmat{S}\dmat{\Sigma}(S_1,S_2)\rmat{S}^\herm\right)\right]\\
    &\overset{(a)}{=} \{\dmat{\Sigma}^\star(S_1,S_2) \} \overset{(b)}{\notin} \set{G}(2),
    \end{split}
\end{equation*}
where $(a)$ follows from Lemma \ref{lem:optimal_unique}, and $(b)$ from Lemma \ref{lem:Sigma_star}, which implies that $\exists \: p(\dmat{S},s_1,s_2)$ such that
\begin{equation*}
\begin{split}
    &\max_{\dmat{\Sigma}\in \set{P}\cap \set{G}(4)}\EE\left[\log\mathrm{det}\left(\dmat{I}+\rmat{S}\dmat{\Sigma}(S_1,S_2)\rmat{S}^\herm\right)\right]\\
    &>\max_{\dmat{\Sigma}\in \set{P}\cap \set{G}(2)}\EE\left[\log\mathrm{det}\left(\dmat{I}+\rmat{S}\dmat{\Sigma}(S_1,S_2)\rmat{S}^\herm\right)\right].
\end{split}
\end{equation*}

\subsection{Proof of Proposition \ref{prop:tighter_bound}}
\label{ssec:proof_tighter_bound}
\begin{proof}[\unskip\nopunct] 
The proof follows by manipulating an off-diagonal entry of the sub-matrices $\dmat{Q}_1,\dmat{Q}_2$ given by \eqref{eq:Q_structure} until $\dmat{Q}$ becomes rank-deficient. We recall that varying these entries has no influence on the achievable rates, provided that the positive semi-definiteness of $\dmat{Q}$ is maintained. Consider the symmetric matrix 
\begin{equation*}
\tilde{\dmat{Q}}(t) \eqdef \dmat{Q}+ t\left[\begin{smallmatrix} \begin{smallmatrix} 0 & 1 \\ 1 & 0 \end{smallmatrix} & \dmat{0} \\
\dmat{0} & \dmat{0}
\end{smallmatrix}\right], \quad t \in \stdset{R},\quad \dmat{Q}\in \stdset{S}_+^d.
\end{equation*}
Let $\lambda_{\min}: \stdset{S}^d \to \stdset{R}$ be the minimum eigenvalue of a Hermitian symmetric matrix (not necessarily positive semi-definite). By definition, $\lambda_{\min}(\tilde{\dmat{Q}}(0)) \geq 0$.  In addition, $\exists t_1 > 0$ such that $\lambda_{\min}(\tilde{\dmat{Q}}(t_1)) < 0$. Furthermore, by the continuity of the map $\lambda_{\min}$ in the matrix entries \cite[Theorem 5.2]{kato2012short}, $\lambda_{\min}(\tilde{\dmat{Q}}(t))$ is a continuous function of $t$. Hence, by the intermediate value theorem, $\exists t_0 \in [0,t_1]$ such that $\lambda_{\min}(\tilde{\dmat{Q}}(t_0)) = 0$, i.e., such that $\tilde{\dmat{Q}}(t_0)$ is positive semi-definite, low rank, and achieves the same rates as $\dmat{Q}$.
\end{proof}

% REFERENCES
\bibliographystyle{IEEEtran}
\bibliography{IEEEabrv,refs} 

% BIO
\begin{IEEEbiographynophoto}{Lorenzo Miretti}(Student Member, IEEE) received the BSc degree in telecommunication engineering from Politecnico di Torino, Turin, Italy, in 2015 and the MSc degree in Communications and Computer Networks Engineering from Politecnico di Torino and T\'el\'ecom ParisTech in 2018, both cum laude. He is currently working towards the Ph.D degree with the Department of Communication Systems, EURECOM, Sophia Antipolis, France. He studies the physical layer of wireless networks, signal processing, and multi-user information theory.
\end{IEEEbiographynophoto}

\begin{IEEEbiographynophoto}{Mari Kobayashi}(Senior Member, IEEE) received the B.E. degree in electrical engineering from Keio University, Yokohama, Japan, in 1999, and the M.S. degree in mobile radio and the Ph.D. degree from École Nationale Supérieure des Télécommunications, Paris, France, in 2000 and 2005, respectively. From November 2005 to March 2007, she was a postdoctoral researcher at the Centre Tecnològic de Telecomunicacions de Catalunya, Barcelona, Spain. In May 2007, she joined the Telecommunications department at Centrale Supélec, Gif-sur-Yvette, France, where she is now a professor. She is the recipient of the Newcom++ Best Paper Award in 2010, and IEEE Comsoc/IT Joint Society Paper Award in 2011, and ICC Best Paper Award in 2019. She was an Alexander von Humboldt Experienced Research Fellow (September 2017- April 2019) and an August-Wihelm Scheer Visiting Professor (August 2019-April 2020) at Technical University of Munich (TUM).  
\end{IEEEbiographynophoto}

\begin{IEEEbiographynophoto}{David Gesbert}(Fellow, IEEE) is Professor and Head of the Communication Systems Department, EURECOM. He obtained the Ph.D degree from Ecole Nationale Superieure des Telecommunications, France, in 1997. From 1997 to 1999 he has been with the Information Systems Laboratory, Stanford University. He was then a founding engineer of Iospan Wireless Inc, a Stanford spin off pioneering MIMO-OFDM (now Intel). Before joining EURECOM in 2004, he has been with the Department of Informatics, University of Oslo as an adjunct professor. D. Gesbert has published about 340 papers and 25 patents, some of them winning 2019 ICC Best Paper Award, 2015 IEEE Best Tutorial Paper Award (Communications Society), 2012 SPS Signal Processing Magazine Best Paper Award, 2004 IEEE Best Tutorial Paper Award (Communications Society), 2005 Young Author Best Paper Award for Signal Proc. Society journals, and paper awards at conferences 2011 IEEE SPAWC, 2004 ACM MSWiM. He has been a Technical Program Co-chair for ICC2017. He was named a Thomson-Reuters Highly Cited Researchers in Computer Science. In 2015, he was awarded the ERC Advanced Grant "PERFUME" on the topic of smart device Communications in future wireless networks. He is a Board member for the OpenAirInterface (OAI) Software Alliance. Since early 2019, he heads the Huawei-funded Chair on Adwanced Wireless Systems Towards 6G Networks. He sits on the Advisory Board of Huawei European Research Institute. In 2020, he was awarded funding by the French Interdisciplinary Institute on Artificial Intelligence for a Chair in the area of AI for the future IoT.
\end{IEEEbiographynophoto}

\begin{IEEEbiographynophoto}{Paul de Kerret}(Member, IEEE) has an Engineering degree from both the French Graduate School, IMT Atlantique, and from the Munich University of Technology through a double-degree program. In 2013, he obtained his doctorate from the French Graduate School, Telecom Paris. From 2015 to 2019, he was a Senior Researcher at EURECOM. In September 2019, he joined the Mantu Artificial Intelligence Laboratory to lead the Research Group and he is currently working as Lead Data Scientist within the Startup Company Greenly where he puts his skills to measure and reduce our carbon footprints. He has been involved in several European collaborative projects on mobile communications and co-presented several tutorials at major IEEE international conferences. He has authored more than 30 articles in the IEEE flagship conferences.
\end{IEEEbiographynophoto}
\end{document}